\documentclass[a4paper,UKenglish]{lipics-report}
 
\usepackage{microtype}

\title{Nested Words for Order-2 Pushdown Systems\footnote{Partially supported 
by LIA InForMeL}}

\author[1]{C.~Aiswarya}
\author[2]{Paul Gastin}
\author[3]{Prakash Saivasan}

\affil[1]{Chennai Mathematical Institute\\\texttt{aiswarya@cmi.ac.in}}
\affil[2]{LSV, ENS-Cachan, CNRS, INRIA, Univ.\ Paris-Saclay, 94235 Cachan, France\\
          \texttt{gastin@lsv.fr}}
\affil[3]{University of Kaiserslautern\\\texttt{saivasan@rhrk.uni-kl.de}}

\authorrunning{C.\ Aiswarya,  P.\ Gastin, and P.\ Saivasan}

\Copyright{C.\ Aiswarya, Paul Gastin, and Prakash Saivasan}


\usepackage{thm-restate}
\usepackage{wrapfig}

\usepackage{xspace}
\usepackage[usenames,dvipsnames]{xcolor}
\usepackage{amsmath,amssymb}
\usepackage{mathtools}
\usepackage{calc}
\usepackage{enumerate}
\usepackage{enumitem}
\usepackage{extarrows}
\usepackage{upgreek}
\usepackage{booktabs}
\usepackage{pifont}
\usepackage{stmaryrd}
\usepackage{rotating}
\usepackage{mathrsfs,mathabx}
\usepackage{cite}
\usepackage{graphicx,url}
\usepackage{pifont}
\usepackage{bbold}
\usepackage{fancyhdr}
\usepackage{caption}

\usepackage[pdflatex,recompilepics=false]{gastex}
\gasset{frame=false,AHangle=30,AHlength=1.4,AHLength=1.6}

\usepackage{ifpdf}
\ifpdf
\usepackage{tikz}
\usepackage{tikz-qtree}
\usetikzlibrary{shapes}
\usetikzlibrary{shapes.multipart}
\usetikzlibrary{positioning}
\usetikzlibrary{decorations}
\usetikzlibrary{decorations.pathmorphing} 
\usetikzlibrary{fit}		
\usetikzlibrary{backgrounds}	
\usetikzlibrary{matrix,arrows}
\usetikzlibrary{calc}
\usetikzlibrary{through}
\usetikzlibrary{decorations.pathreplacing}
\usetikzlibrary{arrows,snakes,automata,backgrounds,petri,positioning,decorations.pathmorphing,fit}

\usepackage[colorinlistoftodos,prependcaption]{todonotes}
\fi

\theoremstyle{definition}

\makeatletter

\renewcommand{\paragraph}{\@startsection{paragraph}{6}{\z@}{2ex}{-0.7em}{\normalsize\bf}}

\makeatother

\renewcommand\subparagraph[1]{\medskip\noindent\textbf{#1}~}

\begin{document}

\maketitle

\begin{abstract}
  We study linear time model checking of collapsible higher-order pushdown
  systems (CPDS) of order 2 (manipulating stack of stacks) against MSO and PDL
  (propositional dynamic logic with converse and loop) enhanced with push/pop
  matching relations.  To capture these linear time behaviours with matchings,
  we propose order-2 nested words.  These graphs consist of a word structure
  augmented with two binary matching relations, one for each order of stack,
  which relate a push with matching pops (or collapse) on the respective stack.
  Due to the matching relations, satisfiability and model checking are
  undecidable.  Hence we propose an under-approximation, bounding the number of
  times an order-1 push can be popped.  With this under-approximation, which 
  still allows unbounded stack height,  we get
  decidability for satisfiability and model checking of both MSO and PDL.
  The problems are ExpTime-Complete for
  PDL.
\end{abstract}


\newcommand{\newtodo}[1]{\todo[linecolor=orange,backgroundcolor=orange!20,bordercolor=orange]{#1}}
\newcommand{\comment}[1] {\color{red}{(** #1 **)}}

\newcommand{\cob}{{\lfloor}}
\newcommand{\ccb}{{\rfloor}}

\newcommand{\Eve}{Eve\xspace}
\newcommand{\Adam}{Adam\xspace}
\newcommand{\snw}{\overline{\nw} }
\newcommand{\snwof}{\textsf{sNWof}}

\newcommand{\Conf} {\mbox{Conf}}
\newcommand{\moves}[1] {\stackrel{#1}{\longrightarrow}}
\newcommand{\move}[1] {\stackrel{#1}{\rightarrow}}

\newcommand{\States} {State}
\newcommand{\Stack} {Stack}
\newcommand{\Trace} {Trace}
\newcommand{\Natural} {\mathbb{N}}
\newcommand{\NestArrow} {\mbox{\curvearrowleft}}

\newcommand{\Source} {Src}
\newcommand{\StackSeq} {\nu}

\newcommand{\Nestify}{\texttt{Nestify}\xspace}

\newcommand{\mypush}{\mathsf{push}_1}
\newcommand{\myPush}{\mathsf{Push}_2}

\newcommand{\grid}{\mathsf{grid}}
\newcommand{\last}{\mathsf{empty}_1}
\newcommand{\ispushone}{\mathsf{ispush}_1}
\newcommand{\ispopone}{\mathsf{ispop}_1}
\newcommand{\ispushtwo}{\mathsf{ispush}_2}
\newcommand{\ispoptwo}{\mathsf{ispop}_2}
\newcommand{\isfirstpushtwo}{\mathsf{isfirstpush}_2}

\newcommand{\SAT}{\mathsf{SAT}}
\newcommand{\MC}{\mathsf{MC}}

\newcommand{\StateSet}{Q}
\newcommand{\StackStateSet}{S}
\newcommand{\StackStateSetBot}{\StackStateSet_\bot}
\newcommand{\TransitionSet}{\Delta}
\newcommand{\InitialState}{q_0}
\newcommand{\FinalStates}{F}
\newcommand{\InitialStackState}{s_0}
\newcommand{\LabelSet}{S}
\newcommand{\Label}{s}
\newcommand{\State}{q}
\newcommand{\state}{\State}
\newcommand{\StackState}{s}
\newcommand{\Alphabet}{\Sigma}
\newcommand{\twocpds}{\textsf{2-CPDS}\xspace}
\newcommand{\hopds}{\mathcal{H}}
\newcommand{\push}{{\color{blue}{\uparrow_1}}}
\newcommand{\Push}{{\color{red}{\uparrow_2}}}
\newcommand{\pushs}[1]{{\color{blue}{\uparrow_1^{#1}}}}
\newcommand{\Pushs}[1]{{\color{red}{\uparrow_2^{#1}}}}
\newcommand{\pop}{{\color{blue}{\downarrow_1}}}
\newcommand{\Pop}{{\color{red}{\downarrow_2}}}
\newcommand{\Pops}[1]{{\color{red}{\downarrow_2^#1}}}
\newcommand{\collapse}{{\color{red}{\Downarrow}}}
\newcommand{\collapses}[1]{{\color{red}{\Downarrow^{#1}}}}
\newcommand{\toptest}{\mathsf{top}}
\newcommand{\OpSet}{\mathsf{Op}}
\newcommand{\op}{\mathsf{op}}
\newcommand{\nop}{\mathsf{Nop}}
\newcommand{\trans}{\tau}
\newcommand{\conf}{C}
\newcommand{\ConfSet}{\mathcal{C}}
\newcommand{\goesto}[1]{\xRightarrow{#1}}
\newcommand{\run}{\rho}
\newcommand{\topone}{\mathsf{top}_1}
\newcommand{\toptwo}{\mathsf{Top}_2}
\newcommand{\NestArr } {\curvearrowleft }
\newcommand{\NWModels} {\mathcal{N}}
\newcommand{\domain} {\mathsf{dom}}
\newcommand{\PushA}[1] { #1^{\downarrow}}
\newcommand{\PopA}[1] { #1^{\uparrow}}
\newcommand{\interpretation} {\mathcal{I}}

\newcommand{\NestRelOne}{\curvearrowright^1}
\newcommand{\NestRelTwo}{\mathrel{\red\curvearrowright^2}}
\newcommand{\tuple}[1]{\langle #1 \rangle}
\newcommand{\nw}{\mathcal{N}}
\newcommand{\NWof}{\textsf{NWof}}
\newcommand{\NWtwo}{\textsf{2-NW}\xspace}

\newcommand{\Lang}{\mathcal{L}}
\newcommand{\sizeof}[1]{|#1|} 
\newcommand{\MSONW}{\ensuremath{\mathsf{MSO}}\xspace}

\newcommand{\godowngrid}{\downarrow}
\newcommand{\goupgrid}{\uparrow}
\newcommand{\goleftgrid}{\leftarrow}
\newcommand{\gorightgrid}{\rightarrow}
\newcommand{\nextpop}{\hookrightarrow}
\newcommand{\PDL}{\ensuremath{\mathsf{LCPDL}}\xspace}
\newcommand{\existspath}[1]{\langle #1 \rangle}
\newcommand{\allpaths}[1]{[#1]}
\newcommand{\existsloop}[1]{\mathsf{Loop}(#1)}
\newcommand{\conc}{\cdot}
\newcommand{\gonext}{\mathop{\rightarrow}}
\newcommand{\goprev}{\leftarrow}
\newcommand{\gopopOne}{\NestRelOne}
\newcommand{\gopushOne}{\curvearrowleft^1}

\newcommand{\gopopTwo}{\NestRelTwo}
\newcommand{\gopushTwo}{\mathrel{\red\curvearrowleft^2}}

\newcommand{\gonextpop}{\mathop{\nextpop}}
\newcommand{\goprevpop}{\mathop{\hookleftarrow}}
\newcommand{\test}[1]{\{#1\}?}
\newcommand{\existsnode}[1]{\mathsf{E} #1}
\newcommand{\forallnodes}[1]{\mathsf{A} #1}
\newcommand{\true}{\mathsf{true}}
\newcommand{\false}{\mathsf{false}}
\newcommand{\Us}{\mathop{\mathsf{U}^\text{s}}}

\newcommand{\twohopds}{\textsf{2-HOPDS}\xspace}
\newcommand{\twonw}{\textsf{2-NW}\xspace}
\newcommand{\twonws}{\textsf{2-NWs}\xspace}
\newcommand{\bound}{\beta}
\newcommand{\pbnw}{\twonw(\bound)}
\newcommand{\poly}{\mathsf{Poly}}

\newcommand{\width}{\mathsf{width}}
\newcommand{\splitwidth}{\mathsf{sw}}

\newcommand{\Aswk}{\ensuremath{\mathcal{A}^{k\mathsf{-sw}}_{\twonw}}\xspace}
\newcommand{\Apopb}{\ensuremath{\mathcal{A}^{\bound}_{\twonw}}\xspace}
\newcommand{\Ahopds}{\ensuremath{\mathcal{A}^{\bound}_{\hopds}}\xspace}
\newcommand{\Aphi}{\ensuremath{\mathcal{A}^{\bound}_{\phi}}\xspace}
\newcommand{\Anotphi}{\ensuremath{\mathcal{A}^{\bound}_{\neg\phi}}\xspace}
\newcommand{\Aword}{\ensuremath{\mathcal{A}^{\bound}_{\mathsf{word}}}\xspace}
\newcommand{\Anesting}{\ensuremath{\mathcal{A}^{\bound}_{\mathsf{nw}}}\xspace}
\newcommand{\STT}{$\mathsf{STT}$\xspace}
\newcommand{\STTs}{$\mathsf{STTs}$\xspace}
\newcommand{\kSTT}{$k$-\STT\xspace}
\newcommand{\kSTTs}{$k$-\STTs\xspace}
\newcommand{\STW}{$\mathsf{STW}$\xspace}
\newcommand{\STWs}{$\mathsf{STWs}$\xspace}
\newcommand{\stt}{\tau}
\newcommand{\add}[3]{\mathop{\mathsf{Add}_{#1,#2}^{#3}}}
\newcommand{\forget}[1]{\mathop{\mathsf{Forget}_#1}}
\newcommand{\rename}[2]{\mathop{\mathsf{Rename}_{#1,#2}}}
\newcommand{\sttunion}{\oplus}
\newcommand{\Nat}{\mathbb{N}}
\newcommand{\sem}[1]{\llbracket #1 \rrbracket}
\newcommand{\dom}{\mathsf{dom}}


\pgfdeclarelayer{bg}    
\pgfsetlayers{bg,main}

\newcommand{\tikzclose}{\includegraphics[scale=1,page=5]{tikz-pics}}
\newcommand{\varcolor}{black}

\begin{gpicture}[ignore, name = gpic:nested-tree ]
  \gasset{Nw=3,Nh=3,Nmr= 4,AHnb=0,Nframe=n}
  \unitlength =2.5
  \node(n1)(0,0) {$1$}
  \node(n2)(-5,-5) {$2$}
  \node(n3)(-10,-10) {$3$}
    \node(n4)(5,-5) {$4$}
  \node(n5)(0,-10) {$5$}
  \node(n6)(-5,-15) {$6$}
  \node(n7)(10,-10) {$7$}
  \node(n8)(15,-15) {$8$}

  \drawedge (n1,n2){}
  \drawedge (n2,n3){}
  \drawedge (n1,n4){}
  \drawedge (n4,n5){}
  \drawedge (n5,n6){}
  \drawedge (n4,n7){}
  \drawedge (n7,n8){}
  
  \gasset{curvedepth=-3, AHnb=1, linecolor=blue}
  \drawedge (n2,n3){}
  \drawedge (n4,n6){}
  \gasset{curvedepth=3}
  \drawedge(n4,n7){}
  \gasset{curvedepth=5}
  \drawedge(n1,n8){}
\end{gpicture}
  
\begin{gpicture}[ignore, name =gpic:tree-to-nw]
  \gasset{Nw=2.5,Nh=2.5,Nframe=n,AHnb=1,AHLength=2,AHlength=1.5}
  \unitlength=.704705882353mm
  \node(n0)(0,0){\normalsize 1}
  \node(n1)(10,0){\normalsize $\bullet$}
  \node(n2)(20,0){\normalsize 2}
  \node(n3)(30,0){\normalsize $\bullet$}
  \node(n4)(40,0){\normalsize 3}
  \node(n5)(50,0){\normalsize $\bullet$}
  \node(n6)(60,0){\normalsize $\bullet$}
  \node(n7)(70,0){\normalsize 4}
  \node(n8)(80,0){\normalsize $\bullet$}
  \node(n9)(90,0){\normalsize 5}
  \node(n10)(100,0){\normalsize $\bullet$}
  \node(n11)(110,0){\normalsize 6}
  \node(n12)(120,0){\normalsize $\bullet$}
  \node(n13)(130,0){\normalsize $\bullet$}
  \node(n14)(140,0){\normalsize 7}
  \node(n15)(150,0){\normalsize 8}
  
   \node(e0)(0,-5){\normalsize $\push$}
  \node(e1)(10,-5){\normalsize $\Push$}
  \node(e2)(20,-5){\normalsize $\push$}
  \node(e3)(30,-5){\normalsize $\Push$}
  \node(e4)(40,-5){\normalsize $\pop$}
  \node(e5)(50,-5){\normalsize $\Pop$}
  \node(e6)(60,-5){\normalsize $\Pop$}
  \node(e7)(70,-5){\normalsize $\push$}
  \node(e8)(80,-5){\normalsize $\Push$}
  \node(e9)(90,-5){\normalsize }
  \node(e10)(100,-5){\normalsize $\Push$}
  \node(e11)(110,-5){\normalsize $\pop$}
  \node(e12)(120,-5){\normalsize $\Pop$}
  \node(e13)(130,-5){\normalsize $\Pop$}
  \node(e14)(140,-5){\normalsize $\pop$}
  \node(e15)(150,-5){\normalsize $\pop$}

  \drawedge[dash={7.875 2.125}0,AHnb=0](n0,n15){}

  \node[Nw=0.1,Nh=0.1](1)(20,8.4){}
  \node[Nw=0.1,Nh=0.1](2)(40,8.4){}
  \drawedge[AHnb=0,linecolor=blue](n2,1){}
  \drawedge[AHnb=0,linecolor=blue](1,2){}
  \drawedge[AHnb=1,linecolor=blue](2,n4){}
  \node[Nw=0.1,Nh=0.1](1)(30,6.4){}
  \node[Nw=0.1,Nh=0.1](2)(50,6.4){}
  \drawedge[AHnb=0,linecolor=red](n3,1){}
  \drawedge[AHnb=0,linecolor=red](1,2){}
  \drawedge[AHnb=1,linecolor=red](2,n5){}
  \node[Nw=0.1,Nh=0.1](1)(10,10.4){}
  \node[Nw=0.1,Nh=0.1](2)(60,10.4){}
  \drawedge[AHnb=0,linecolor=red](n1,1){}
  \drawedge[AHnb=0,linecolor=red](1,2){}
  \drawedge[AHnb=1,linecolor=red](2,n6){}
  \node[Nw=0.1,Nh=0.1](2)(110,10.4){}
  \drawedge[AHnb=1,linecolor=blue](2,n11){}
  \node[Nw=0.1,Nh=0.1](1)(100,6.4){}
  \node[Nw=0.1,Nh=0.1](2)(120,6.4){}
  \drawedge[AHnb=0,linecolor=red](n10,1){}
  \drawedge[AHnb=0,linecolor=red](1,2){}
  \drawedge[AHnb=1,linecolor=red](2,n12){}
  \node[Nw=0.1,Nh=0.1](1)(80,8.4){}
  \node[Nw=0.1,Nh=0.1](2)(130,8.4){}
  \drawedge[AHnb=0,linecolor=red](n8,1){}
  \drawedge[AHnb=0,linecolor=red](1,2){}
  \drawedge[AHnb=1,linecolor=red](2,n13){}
  \node[Nw=0.1,Nh=0.1](1)(70,10.4){}
  \node[Nw=0.1,Nh=0.1](2)(140,10.4){}
  \drawedge[AHnb=0,linecolor=blue](n7,1){}
  \drawedge[AHnb=0,linecolor=blue](1,2){}
  \drawedge[AHnb=1,linecolor=blue](2,n14){}
  \node[Nw=0.1,Nh=0.1](1)(0,12.4){}
  \node[Nw=0.1,Nh=0.1](2)(150,12.4){}
  \drawedge[AHnb=0,linecolor=blue](n0,1){}
  \drawedge[AHnb=0,linecolor=blue](1,2){}
  \drawedge[AHnb=1,linecolor=blue](2,n15){}
\end{gpicture}  

\begin{gpicture}[ignore, name=gpic:twocpds]
  \gasset{Nw=5,Nh=5,Nmr= 4,AHnb=1}
  
  \node[Nmarks=i,iangle=90](n0)(0,0) {$0$}
  \node(n1)(24,8) {$1$}
  \node(n2)(24,-8) {$2$}
  
  \drawedge[curvedepth=1.5,ELpos=40](n1,n0){$\pushs a$}
  \drawedge[curvedepth=1.5](n0,n1){$\Push$}
  \drawedge (n1,n2){$\pushs b$}
  \drawedge (n2,n0){$\collapse$}
  \drawloop[loopdiam=4,loopangle=270,ELpos=75](n2) {$\pop$}
  
\end{gpicture}

\begin{gpicture}[ ignore,name =gpic:twocpdsNoCollapse]
  \gasset{Nw=5,Nh=5,Nmr= 4,AHnb=1}
  
  \node[Nmarks=if,fangle=90, iangle=-90](n0)(0,0) {$0$}
  \node(n1)(24,8) {$1$}
  \node(n2)(24,-8) {$2$}
  
  \drawedge[curvedepth](n0,n1){$\pushs a$}
  \drawedge (n1,n2){$\Push$}
  \drawedge (n2,n0){$\toptest(\bot)?$;$\Pop$}
  \drawloop[loopdiam=4,loopangle=-90](n2) {$\toptest(a)?$;$\pop$}
  
\end{gpicture}

\begin{gpicture}[ignore,name=gpic:nw2]
  \gasset{Nw=4,Nh=4,Nframe=n,AHnb=1,AHLength=1.5,AHlength=1}
  \unitlength=0.63333333333333mm
  \node(n0)(0,0){\small 0}
  \node(n1)(10,0){\small 1}
  \node(n2)(20,0){\small 2}
  \node(n3)(30,0){\small 3}
  \node(n4)(40,0){\small 4}
  \node(n5)(50,0){\small 5}
  \node(n6)(60,0){\small 6}
  \node(n7)(70,0){\small 7}
  \node(n8)(80,0){\small 8}
  \node(n9)(90,0){\small 9}
  \node(n10)(100,0){\small 10}
  \node(n11)(110,0){\small 11}
  \node(n12)(120,0){\small 12}
  \node(n13)(130,0){\small 13}
  \node(n14)(140,0){\small 14}
  \node(n15)(150,0){\small 15}
  \node(n16)(160,0){\small 16}
  \node(n17)(170,0){\small 17}
  \node(n18)(180,0){\small 18}
  \node(n19)(190,0){\small 19}
  \node(n20)(200,0){\small 20}
  \node(n21)(210,0){\small 21}
  
  \node(e0)(0,-6){\small $\pushs a$}
  \node(e1)(10,-6){\small $\Push$}
  \node(e2)(20,-6){\small $\pop$}
  \node(e3)(30,-6){\small $\Pop$}
  \node(e4)(40,-6){\small $\pushs a$}
  \node(e5)(50,-6){\small $\Push$}
  \node(e6)(60,-6){\small $\pop$}
  \node(e7)(70,-6){\small $\pop$}
  \node(e8)(80,-6){\small $\Pop$}
  \node(e9)(90,-6){\small $\pushs a$}
  \node(e10)(100,-6){\small $\Push$}
  \node(e11)(110,-6){\small $\pop$}
  \node(e12)(120,-6){\small $\pop$}
  \node(e13)(130,-6){\small $\pop$}
  \node(e14)(140,-6){\small $\Pop$}
  \node(e15)(150,-6){\small $\pushs a$}
  \node(e16)(160,-6){\small $\Push$}
  \node(e17)(170,-6){\small $\pop$}
  \node(e18)(180,-6){\small $\pop$}
  \node(e19)(190,-6){\small $\pop$}
  \node(e20)(200,-6){\small $\pop$}
  \node(e21)(210,-6){\small $\Pop$}
  
  \drawedge[dash={6 4}0,AHnb=0](n0,n21){}
  
  \node[Nw=0.1,Nh=0.1](2)(20,14.4){}
  \drawedge[AHnb=1,linecolor=blue](2,n2){}
  \node[Nw=0.1,Nh=0.1](1)(10,6.4){}
  \node[Nw=0.1,Nh=0.1](2)(30,6.4){}
  \drawedge[AHnb=0,linecolor=red](n1,1){}
  \drawedge[AHnb=0,linecolor=red](1,2){}
  \drawedge[AHnb=1,linecolor=red](2,n3){}
  \node[Nw=0.1,Nh=0.1](2)(60,12.4){}
  \drawedge[AHnb=1,linecolor=blue](2,n6){}
  \node[Nw=0.1,Nh=0.1](2)(70,14.4){}
  \drawedge[AHnb=1,linecolor=blue](2,n7){}
  \node[Nw=0.1,Nh=0.1](1)(50,6.4){}
  \node[Nw=0.1,Nh=0.1](2)(80,6.4){}
  \drawedge[AHnb=0,linecolor=red](n5,1){}
  \drawedge[AHnb=0,linecolor=red](1,2){}
  \drawedge[AHnb=1,linecolor=red](2,n8){}
  \node[Nw=0.1,Nh=0.1](2)(110,10.4){}
  \drawedge[AHnb=1,linecolor=blue](2,n11){}
  \node[Nw=0.1,Nh=0.1](2)(120,12.4){}
  \drawedge[AHnb=1,linecolor=blue](2,n12){}
  \node[Nw=0.1,Nh=0.1](2)(130,14.4){}
  \drawedge[AHnb=1,linecolor=blue](2,n13){}
  \node[Nw=0.1,Nh=0.1](1)(100,6.4){}
  \node[Nw=0.1,Nh=0.1](2)(140,6.4){}
  \drawedge[AHnb=0,linecolor=red](n10,1){}
  \drawedge[AHnb=0,linecolor=red](1,2){}
  \drawedge[AHnb=1,linecolor=red](2,n14){}
  \node[Nw=0.1,Nh=0.1](1)(150,8.4){}
  \node[Nw=0.1,Nh=0.1](2)(170,8.4){}
  \drawedge[AHnb=0,linecolor=blue](n15,1){}
  \drawedge[AHnb=0,linecolor=blue](1,2){}
  \drawedge[AHnb=1,linecolor=blue](2,n17){}
  \node[Nw=0.1,Nh=0.1](1)(90,10.4){}
  \node[Nw=0.1,Nh=0.1](2)(180,10.4){}
  \drawedge[AHnb=0,linecolor=blue](n9,1){}
  \drawedge[AHnb=0,linecolor=blue](1,2){}
  \drawedge[AHnb=1,linecolor=blue](2,n18){}
  \node[Nw=0.1,Nh=0.1](1)(40,12.4){}
  \node[Nw=0.1,Nh=0.1](2)(190,12.4){}
  \drawedge[AHnb=0,linecolor=blue](n4,1){}
  \drawedge[AHnb=0,linecolor=blue](1,2){}
  \drawedge[AHnb=1,linecolor=blue](2,n19){}
  \node[Nw=0.1,Nh=0.1](1)(0,14.4){}
  \node[Nw=0.1,Nh=0.1](2)(200,14.4){}
  \drawedge[AHnb=0,linecolor=blue](n0,1){}
  \drawedge[AHnb=0,linecolor=blue](1,2){}
  \drawedge[AHnb=1,linecolor=blue](2,n20){}
  \node[Nw=0.1,Nh=0.1](1)(160,6.4){}
  \node[Nw=0.1,Nh=0.1](2)(210,6.4){}
  \drawedge[AHnb=0,linecolor=red](n16,1){}
  \drawedge[AHnb=0,linecolor=red](1,2){}
  \drawedge[AHnb=1,linecolor=red](2,n21){}
\end{gpicture}

\begin{gpicture}[ignore, name=gpic:nw1]
  \gasset{Nw=4,Nh=4,Nframe=n,AHnb=1,AHLength=1.5,AHlength=1}
  \unitlength=.726315789474mm
  \node(n0)(0,0){\small 0}
  \node(n1)(10,0){\small 1}
  \node(n2)(20,0){\small 2}
  \node(n3)(30,0){\small 3}
  \node(n4)(40,0){\small 4}
  \node(n5)(50,0){\small 5}
  \node(n6)(60,0){\small 6}
  \node(n7)(70,0){\small 7}
  \node(n8)(80,0){\small 8}
  \node(n9)(90,0){\small 9}
  \node(n10)(100,0){\small 10}
  \node(n11)(110,0){\small 11}
  \node(n12)(120,0){\small 12}
  \node(n13)(130,0){\small 13}
  \node(n14)(140,0){\small 14}
  \node(n15)(150,0){\small 15}
  \node(n16)(160,0){\small 16}
  \node(n17)(170,0){\small 17}
  
  \node(e0)(0,-5){\small $\Push$}
  \node(e1)(10,-5){\small $\pushs a$}
  \node(e2)(20,-5){\small $\Push$}
  \node(e3)(30,-5){\small $\pushs a$}
  \node(e4)(40,-5){\small $\Push$}
  \node(e5)(50,-5){\small $\pushs a$}
  \node(e6)(60,-5){\small $\Push$}
  \node(e7)(70,-5){\small $\pushs b$}
  \node(e8)(80,-5){\small $\pop$}
  \node(e9)(90,-5){\small $\pop$}
  \node(e10)(100,-5){\small $\collapse$}
  \node(e11)(110,-5){\small $\Push$}
  \node(e12)(120,-5){\small $\pushs a$}
  \node(e13)(130,-5){\small $\Push$}
  \node(e14)(140,-5){\small $\pushs b$}
  \node(e15)(150,-5){\small $\pop$}
  \node(e16)(160,-5){\small $\pop$}
  \node(e17)(170,-5){\small $\collapse$}

  \drawedge[dash={6 4}0,AHnb=0](n0,n17){}

  \node[Nw=0.1,Nh=0.1](1)(70,6.4){}
  \node[Nw=0.1,Nh=0.1](2)(80,6.4){}
  \drawedge[AHnb=0,linecolor=blue](n7,1){}
  \drawedge[AHnb=0,linecolor=blue](1,2){}
  \drawedge[AHnb=1,linecolor=blue](2,n8){}
  \node[Nw=0.1,Nh=0.1](1)(50,10.4){}
  \node[Nw=0.1,Nh=0.1](2)(90,10.4){}
  \drawedge[AHnb=0,linecolor=blue](n5,1){}
  \drawedge[AHnb=0,linecolor=blue](1,2){}
  \drawedge[AHnb=1,linecolor=blue](2,n9){}
  \node[Nw=0.1,Nh=0.1](1)(60,8.4){}
  \node[Nw=0.1,Nh=0.1](2)(100,8.4){}
  \drawedge[AHnb=0,linecolor=red](n6,1){}
  \drawedge[AHnb=0,linecolor=red](1,2){}
  \drawedge[AHnb=1,linecolor=red](2,n10){}
  \node[Nw=0.1,Nh=0.1](1)(40,12.4){}
  \node[Nw=0.1,Nh=0.1](2)(100,12.4){}
  \drawedge[AHnb=0,linecolor=red](n4,1){}
  \drawedge[AHnb=0,linecolor=red](1,2){}
  \drawedge[AHnb=1,linecolor=red](2,n10){}
  \node[Nw=0.1,Nh=0.1](1)(20,14.4){}
  \node[Nw=0.1,Nh=0.1](2)(100,14.4){}
  \drawedge[AHnb=0,linecolor=red](n2,1){}
  \drawedge[AHnb=0,linecolor=red](1,2){}
  \drawedge[AHnb=1,linecolor=red](2,n10){}
  \node[Nw=0.1,Nh=0.1](1)(140,6.4){}
  \node[Nw=0.1,Nh=0.1](2)(150,6.4){}
  \drawedge[AHnb=0,linecolor=blue](n14,1){}
  \drawedge[AHnb=0,linecolor=blue](1,2){}
  \drawedge[AHnb=1,linecolor=blue](2,n15){}
  \node[Nw=0.1,Nh=0.1](1)(120,10.4){}
  \node[Nw=0.1,Nh=0.1](2)(160,10.4){}
  \drawedge[AHnb=0,linecolor=blue](n12,1){}
  \drawedge[AHnb=0,linecolor=blue](1,2){}
  \drawedge[AHnb=1,linecolor=blue](2,n16){}
  \node[Nw=0.1,Nh=0.1](1)(130,8.4){}
  \node[Nw=0.1,Nh=0.1](2)(170,8.4){}
  \drawedge[AHnb=0,linecolor=red](n13,1){}
  \drawedge[AHnb=0,linecolor=red](1,2){}
  \drawedge[AHnb=1,linecolor=red](2,n17){}
  \node[Nw=0.1,Nh=0.1](1)(110,12.4){}
  \node[Nw=0.1,Nh=0.1](2)(170,12.4){}
  \drawedge[AHnb=0,linecolor=red](n11,1){}
  \drawedge[AHnb=0,linecolor=red](1,2){}
  \drawedge[AHnb=1,linecolor=red](2,n17){}
  \node[Nw=0.1,Nh=0.1](1)(40,12.4){}
  \node[Nw=0.1,Nh=0.1](2)(170,12.4){}
  \node[Nw=0.1,Nh=0.1](1)(20,14){}
  \node[Nw=0.1,Nh=0.1](2)(170,14){}
  \node[Nw=0.1,Nh=0.1](1)(0,16.4){}
  \node[Nw=0.1,Nh=0.1](2)(170,16.4){}
  \drawedge[AHnb=0,linecolor=red](n0,1){}
  \drawedge[AHnb=0,linecolor=red](1,2){}
  \drawedge[AHnb=1,linecolor=red](2,n17){}
\end{gpicture}
 

\section{Introduction}\label{sec:introduction}

The study of higher-order pushdown systems (HOPDS), has been a prominent line of
research \cite{Aho67,Maslov70,Greibach70,KnapikNU:2002, CarayolW03,Broadbent12,
TCachat03, Ong:2006, Ong-survey:2013, Ong-survey:2015}.  HOPDS are equipped with
a stack (order-1 stack, classical pushdown), or a stack of stacks (order-2
stack), or a stack of stacks of stacks (order-3 stack) and so on.  They
naturally extend the classical pushdown systems to higher orders, and at the
same time they can be used to model higher-order functions
\cite{HagueMOS:2008,Ong-survey:2013,Ong-survey:2015} which is a feature
supported by many widely-used programming languages like scala, python, etc.
\cite{wikipedia:higher-order-functions}.

An extension of HOPDS known as Collapsible HOPDS (CPDS) characterizes recursive
schemes \cite{HagueMOS:2008}.  In this article we focus on CPDS of order 2
(denoted $\twocpds$).  Classically, HOPDS and CPDS can be thought as generating
a set of words (linear behaviour), or a tree (branching behaviour) or a
configuration graph \cite{Ong-survey:2015}.  Here we consider yet another way of
understanding them, as generators of linear behaviours with matching relations,
like in nested words \cite{AlurM09}.
We call these structures order-2 nested words (\twonws).  They are essentially
words augmented with two binary relations --- an {\blue order-1 nesting
relation} and an {\red order-2 nesting relation} which link matching pushes and
pops or collapses of the stack of the respective order.  See
Figure~\ref{fig:ex2nw} depicting a \twonw $\nw_1$ with collapse and 
another one $\nw_2$ without collapse.

\begin{figure}
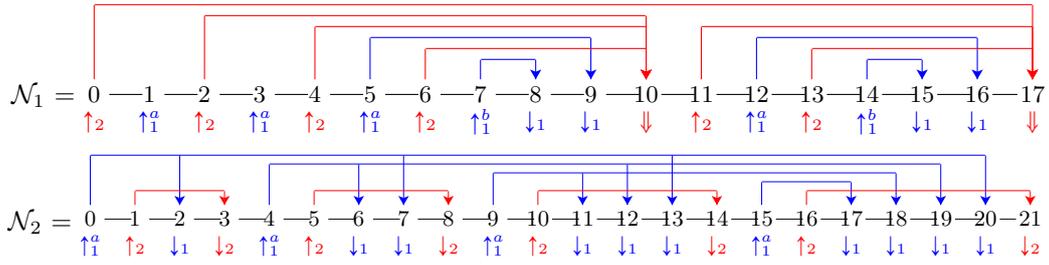

\centering
{
$\nw_1$ = \raisebox{-4.5mm}{\gusepicture{gpic:nw1}}\\[2mm]
$\nw_2$ = \raisebox{-3.5mm}{\gusepicture[scale=.93]{gpic:nw2}}
}
\caption{Two \twonws along with the sequence of operations generating them.
$\push$ means order-1 push, $\pop$ means order-1 pop, $\Push$ means order-2
push, $\Pop$ means order-2 pop, and $\collapse$ means collapse.
}\label{fig:ex2nw}
\end{figure}

We provide a characterisation of the push matching a given pop or collapse, by a
context-free grammar. This allows us to compute the nesting edges, given a
sequence of operations.  Based on it, we have linear time algorithm, \Nestify,
for doing the same.
Our tool \Nestify, accessible at \url{http://www.lsv.fr/~gastin/hopda},
generates a \twonw representation in several formats, including pdf pictures as
in Example~\ref{fig:ex2nw}.
 
We propose propositional dynamic logic with loop and converse ($\PDL$) and
monadic second-order logic ($\MSONW$) over order-2 nested words to specify
properties of \twocpds.  \PDL is a navigating logic which can walk in the
order-2 nested word by moving along the nesting edges and the linear edges.  It
is powerful enough to subsume usual temporal logic operators.
These logics are very expressive since they can use the nesting relations.  We
show that the satisfiability checking of these logics, and model checking of
$\twocpds$ against them
are undecidable.  The reason is that we can interpret grids in \twonws using
these logics.
  
  Our results are quite surprising, since they differ from the established results under classical semantics. Strikingly
\begin{itemize}[nosep]
  \item Model checking and satisfiability problems of $\MSONW$ and $\PDL$ under $\twonw$ semantics
  turn out to be
  undecidable, even when the $\twocpds$ is non-collapsible.  Further, in our
  undecidability proof, the height of the order-2 stack is bounded by 2.  On the
  other hand, MSO over non-collapsible 2-HOPDS under classical semantics is decidable
  \cite{Caucal:2002,KnapikNU:2002}.

  \item The satisfiability and model checking problems described above can be
  reduced to that of non-collapsible $\twocpds/\twonws$.  Contrast this with the
  fact that, when considering HOPDS under classical semantics
  collapse is strictly more powerful.
  \item In \cite{Ong:2006}, Ong showed that $\mu$-calculus over $\twocpds$
  (classical semantics) is decidable.  In the case of $\twonws$, $\PDL$ is
  undecidable even over non-collapsible $2$-HOPDS.
\end{itemize}

Inspired by the success of under-approximation techniques in verification of
otherTturing powerful settings like multi-pushdown systems, message passing
systems etc., we propose an under-approximation for \twocpds, to confront the
undecidability.  This under-approximation, called bounded-pop, bounds the number
of order-1 pops that a push can have. Notice that this does not bound the 
height of order-1 or order-2 stacks.
With this restriction we gain decidability for satisfiability and model checking
of $\MSONW$.  For $\PDL$, we show that these problems are
$\textsc{ExpTime}$-Complete.
  
We establish decidability by showing that bounded-pop \twonws can be
interpreted over trees.  Towards a tree-interpretation, we first lift the notion
of split-width \cite{CGN12,cyriac-phd2014,AGN-atva14} to \twonws, and show that
bounded-pop \twonws have a bound on split-width.  Split-width was first
introduced in \cite{CGN12} for MSO-decidability of multiply-nested words.  It
was later generalised to message sequence charts with nesting (also concurrent
behaviours with matching) \cite{cyriac-phd2014,AGN-atva14}.
\textcolor{black}{Bounded split-width \twonws have bounded (special) tree-width
\cite{Courcelle10}, and hence bounded-pop \twonws can be effectively interpreted
over special tree-terms.}

\section{Order-2 pushdown systems with collapse}\label{sec:2cpds}
  
\subparagraph{Order-2 stacks.} An order-2 stack is a stack of stacks.  In a
collapsible order-2 stack, the stack symbols may in addition contain a pointer
to some stack in the stack of stacks.  Let $\LabelSet$ be a finite set of stack
symbols.  An order-2 stack is of the form $W=[[u_1][u_2] \dots [u_n]]$, where
each $[u_i]$ is a stack over $\LabelSet$ with collapse-pointers to stacks below.
Thus we may see the contents $u_i$ of the $i$-th stack as a word from $\LabelSet
\times \{1, 2, \dots, {i-1}\}$ where the second component of an entry indicates
the index of the stack to which the collapse-link points.  The empty order-2
stack is denoted $[[]]$, where the order-2 stack contains an empty stack.  We
have the following operations on order-2 stacks:
\begin{itemize}[nosep]
  \item $\Push$:  duplicates the topmost stack in the order-2 stack. 
  That is, $\Push([[u_1][u_2]\dots[u_n]]) = [[u_1][u_2]\dots[u_n][u_{n+1}]]$ with $u_{n+1} = u_n$. 

  \item $\Pop$:  pops the topmost stack from the order-2 stack. 
  That is, $\Pop([[u_1][u_2]\dots[u_n]]) = [[u_1][u_2]\dots[u_{n-1}]]$.
  Notice that $\Pop([[u_1]])$ is undefined.  

  \item $\pushs s$:  pushes a symbol $s$ to the top of the topmost stack. Further the pushed symbol contains a ``collapse link'' to the topmost but one stack of the order-2 stacks.\footnote{Collapse links to order-1 stack symbols, or pushes without collapse links are not considered for simplicity. These are, however, easy to simulate thanks to $\pop$ and the top-test.}
  $\pushs s([[u_1][u_2]\dots[u_n]]) = [[u_1][u_2]\dots[u_n(s,n-1)]]$.
  
  \item $\pop$:  removes the topmost element from the topmost stack. 
  $\pop([[u_1][u_2]\dots[u_n]]) = [[u_1][u_2]\dots[u'_n]]$, if $u_n = u'_n(s,i)$.
  Notice that $\pop([[u_1][u_2]\dots[]])$ is undefined.  
  
  \item $\collapse$:
  the collapse operation pops the stacks in the order-2 stack until the stack
  pointed-to by the link in the topmost symbol of the previously topmost stack
  becomes the topmost stack.  $\collapse([[u_1][u_2]\dots[u_n]]) =
  [[u_1][u_2]\dots[u_i]]$, if $u_n = u'_n(s,i)$.  
  
  \item $\toptest(s)$: checks if the topmost symbol of the topmost stack is $s$.
  
  Hence, $\toptest(s)([[u_1][u_2]\dots[u_n]]) = [[u_1][u_2]\dots[u_n]]$, if $u_n
  = u'_n(s,i)$.  It is undefined otherwise.  Also, we can check whether the
  topmost stack is empty by $\toptest(\bot)$.
\end{itemize}
The above defined operations form the set $\OpSet(\LabelSet)$.

\begin{wrapfigure}[14]{r}[.5\columnsep]{32mm}
  \gusepicture[scale=.9]{gpic:twocpds}
  \medskip
  \gusepicture[scale=.9]{gpic:twocpdsNoCollapse}
  \label{ex:twocpdshalfgrid}
\end{wrapfigure}

\subparagraph{\twocpds} is a finite state system over a finite alphabet
$\Alphabet$ equipped with an order-2 stack.  Formally it is a tuple $\hopds =
(\StateSet, \StackStateSet, \TransitionSet, \InitialState, \FinalStates)$ where
$\StateSet$ is the finite set of states, $\StackStateSet$ is the set of stack
symbols/labels, $\InitialState$ is the initial state, $\FinalStates$ is the set
of accepting states, and $\TransitionSet \subseteq \StateSet \times \Alphabet
\times \OpSet(\StackStateSet) \times \StateSet$ is the set of transitions.  On
the right, we have a \twocpds $\hopds_1$ with collapse and a second one
$\hopds_2$ without collapse operations.

A \emph{configuration} is a pair $\conf = (\State, W)$ where $\State \in
\StateSet$ is a state and $W$ is an order-2 stack.  
The \emph{initial} configuration $\conf_0 =(\InitialState,
[[]])$.
 A configuration $\conf = (\State, W)$ is \emph{accepting}
if $\State \in \FinalStates$. 
We write $\conf
\goesto{\trans} \conf' $ for configurations
$\conf = (\State,W)$,  $\conf' = (\State', W')$ and  transition
$\trans = (\State, a, \op(\Label), \State')$,  if $W' = \op(\Label)(W)$.
A \emph{run} $\run$ of $\hopds$ is an alternating sequence of configurations and
transitions, starting from the initial configuration, and conforming to the
relation $\goesto{}$, i.e., $\run = \conf_0 \goesto{\trans_1} \conf_1
\goesto{\trans_2} \conf_2 \ldots \goesto{\trans_n} \conf_n$.
We say that $\run$ is an \emph{accepting run} if $\conf_n$ is accepting.

Next, we aim at understanding the linear behaviours of $\twocpds$ as order-2
nested words ($\twonw$).  For instance, $\nw_1$ and $\nw_2$ of
Figure~\ref{fig:ex2nw} are generated, respectively, by $\hopds_1$ and $\hopds_2$
above.  We give the formal definition below.

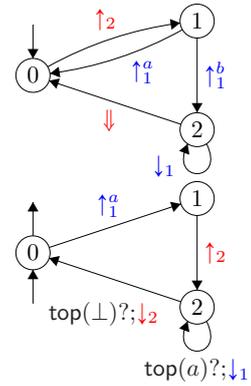
\subparagraph{Order-2 nested words ($\twonw$).} %
We propose words augmented with nesting relations to capture the behaviours of
\twocpds, analogous to nested words for pushdown systems.  We have two nesting
relations $\NestRelOne$ and $\NestRelTwo$, for order-1 stacks and order-2 stacks
respectively.  For each position $j$ executing $\pop$, we find the 
position $i$ at which the popped symbol was pushed and we link these matching 
positions with $i\NestRelOne j$.  Notice that since a stack may be duplicated multiple times, a
$\push$ event may have multiple $\pop$ partners.  For instance, in $\nw_2$
above, the $\push$ at position 4 is matched by the $\pop$ at positions 6,12,19.
Similarly, $\NestRelTwo$ links $\Push$ events with matching $\Pop$ events.
Every $\Push$ event may have at most one $\NestRelTwo$ partner, since each
pushed stack is popped at most once.  A $\collapse$ event is seen as popping
several stacks in one go, hence, several $\Push$ events may be linked to single
$\collapse$ event by $\NestRelTwo$ relation.  For instance, in $\nw_1$ above 
the collapse at position 11 pops the stacks pushed at positions 2,4,6.
Thus an order-2 nested word ($\NWtwo$) over an alphabet $\Alphabet$ is a tuple
$\nw = \tuple{w, \NestRelOne, \NestRelTwo}$ where $w$ is a word over
$\Alphabet$, and $\NestRelOne$ and $\NestRelTwo$ are binary relations over
positions of $w$.  

The language of a $\twocpds$ over an alphabet $\Alphabet$ is a set of $\twonw$
over $\Alphabet$ generated by accepting runs.  It is denoted $\Lang(\hopds)$. When depicting $\twonw$s, we
sometimes do not indicate the labelling by the finite alphabet, but often
indicates the type of the stack operation.  When we express the letter of the
alphabet and the operation, we just write them next to each other.  For example
$a\pushs s$ would mean that the label is $a$, and that position performs $\pushs
s$.

\medskip
Given a generating sequence $\op_0\op_1\cdots\op_n\in\OpSet^+$ of operations,
either it is not valid, or there are unique $\NestRelOne$ and $\NestRelTwo$
which conform to the order-2 stack policy.  To characterize these relations, we
define the position $\mypush(n)$ at which the current (after $\op_n$) top stack
symbol was pushed and the position $\myPush(n)$ at which the current top order-1
stack was pushed/duplicated.  We let $\mypush(n)=-1$ if the top (order-1) stack
is empty.  We let $\myPush(n)=-1$ if the order-2 stack contains only one order-1
stack.
For instance, with the sequence generating $\nw_2$ we have
$\myPush(5)=5=\myPush(7)$, $\myPush(4)=-1=\myPush(8)$,
$\mypush(4)=4=\mypush(5)=\mypush(8)$ and $\mypush(0)=0=\mypush(3)=\mypush(6)$
and $\mypush(2)=-1=\mypush(7)$. Also, in the \twonw $\nw_1$ we have 
$\mypush(11)=1$ and $\myPush(11)=0$.

Surprisingly, $\mypush$ and $\myPush$ can be characterized by a context-free
grammar.  We denote by $L_1$ and $L_2$ the languages defined by the
non-terminals $S_1$ and $S_2$ of the following grammar: 

\medskip\noindent\hfil
$\begin{array}{rcl}
  S_1 &\rightarrow& \push \mid S_1 \Push 
  \mid S_1 S_1 \pop 
  \mid S_1 S_2 \Pop 
  \mid S_1 S_2 S_1 \collapse 
  \\
  S_2 &\rightarrow& \Push \mid S_2 \push \mid S_2 \pop 
  \mid S_2 S_2 \Pop 
  \mid S_2 S_2 S_1 \collapse \,.
\end{array}$

\begin{restatable}{proposition}{rstpropmypush}\label{prop:mypush}
  Let $\op_0\op_1\cdots\op_n\in\OpSet^+$ be a valid push/pop/collapse sequence. 
  Then, for all $0\leq i\leq j\leq n$ we have (proof in 
  Appendix~\ref{app:characterisation})
  \begin{enumerate}[nosep,label={P\arabic*.},ref={P\arabic*},leftmargin=7mm,align=right]
    \item\label{item:P1} $\mypush(j)=i$ iff $\op_i\cdots\op_j\in L_1$,
  
    \item\label{item:P2} $\myPush(j)=i$ iff $\op_i\cdots\op_j\in L_2$,

    \item\label{item:P3} $\mypush(j)=-1$ iff $\op_k\cdots\op_j\notin L_1$ for 
    all $0\leq k\leq j$,
  
    \item\label{item:P4} $\myPush(j)=-1$ iff $\op_k\cdots\op_j\notin L_2$ for 
    all $0\leq k\leq j$.
  \end{enumerate}
\end{restatable}

This  characterization will be crucial in the rest of the paper,  to  justify  correctness of both formulas in Section~\ref{sec:logic}, and also tree-automata constructions in our decision procedure.  Also, it yields a linear time algorithm \Nestify (\url{http://www.lsv.fr/~gastin/hopda}).


\section{PDL and MSO over order-2 nested words}\label{sec:logic}

We introduce two logical formalisms for specifications over \twonw.  The first
one is propositional dynamic logic which essentially navigates
through the edges of a \twonw, checking positional properties on the way.  The
second one is the yardstick monadic second-order logic, which extends MSO over
words with the $\NestRelOne$ and $\NestRelTwo$ binary relations.

Propositional dynamic logic was originally introduced in \cite{FisL79} to study
the branching behaviour of programs.  Here we are not interested in the
branching behaviour.  Instead we study the linear time behaviours (words)
enriched with the nesting relations.  Since these are graphs, we take advantage
of the path formulas of PDL based on regular expressions to navigate in the
\twonws.  This is in the spirit of
\cite{HeTh99,BKM-lmcs10,AGN-atva14,AG-fsttcs14} where PDL was used to specify
properties of graph structures such as message sequence charts or multiply
nested words.

\subparagraph{Propositional Dynamic Logic with converse and loop (\PDL)} can
express properties of nodes (positions) as boolean combinations of the existence
of paths and loops.  Paths are built using regular expressions over the edge
relations (and their converses) of order-2 nested words.
The syntax of the node formulas $\varphi$ and path formulas $\pi$ of \PDL are
given by
\begin{eqnarray*}
\varphi & := & a \mid \varphi \vee \varphi \mid\neg \varphi \mid \existspath{\pi}\varphi \mid \existsloop{\pi}\\
\pi &:= & \test{\varphi} \mid {\gonext} \mid {\goprev} \mid {\gopopOne} \mid {\gopushOne} \mid {\gopopTwo} \mid {\gopushTwo} 
\mid \pi \conc \pi \mid \pi + \pi \mid \pi^\ast
\end{eqnarray*}
where $a \in \Alphabet$. The node formulas are evaluated on positions of an order-2 nested word, whereas
path formulas are evaluated on pairs of positions. 
We give the semantics below ($i,j, i', j'$ vary
over positions of a \twonw $\nw= \tuple{w, \NestRelOne,
\NestRelTwo}$): 
\begin{eqnarray*}
\nw, i \models a & \text{ if } & i\text{th letter of $w$ is $a$}\\
\nw, i \models \varphi_1 \vee \varphi_2 & \text{ if } & \nw, i \models \varphi_1  \text{ or }  \nw, i \models \varphi_2\\
\nw, i \models \neg \varphi & \text{ if } & \text{ it is not the case that }\nw, i \models \varphi\\
\nw, i \models \existspath{\pi}\varphi & \text{ if } & \nw, i, j \models \pi \text{ and } \nw, j \models \varphi \text{ for some j}\\
\nw, i \models \existsloop{\pi} & \text{ if } & \nw, i, i \models \pi \\
\nw, i, j \models \test{\varphi} & \text{ if } & i = j \text{ and } \nw, i \models \varphi\\
\nw, i, j \models {\gonext} & \text{ if } & j \text{ is the successor position of $i$ in the word $w$}\\
\nw, i, j \models {\goprev} & \text{ if } & \nw, j, i \models \gonext\\
\nw, i, j \models {\gopopOne} & \text{ if } & i\NestRelOne j \text{ in the $\twonw$ } \nw \\
\nw, i, j \models {\gopushOne} & \text{ if } & \nw, j, i  \models {\gopopOne} \\
\nw, i, j \models {\gopopTwo} & \text{ if } & i\NestRelTwo j \text{ in the $\twonw $ } \nw\\
\nw, i, j \models {\gopushTwo} & \text{ if } & \nw, j, i  \models {\gopopTwo} \\
\nw, i, j \models {\pi_1 \conc \pi_2} & \text{ if } & \text{ there is a position $k$ such that } \nw, i, k \models \pi_1 \text{ and } \nw, k, j \models {\pi_2}\\
\nw, i, j \models {\pi_1 + \pi_2} & \text{ if } &  \nw, i, j \models \pi_1 \text{ or } \nw, i, j \models {\pi_2}\\
\nw, i, j \models \pi^\ast & \text{ if } &  \text{there exist positions $i_1, \dots ,i_n$ for some $n \ge 1$}\\
&&\text{such that $i = i_1$, $j = i_n$ and $\nw, i_m, i_{m+1} \models \pi$ for all $1 \le m < n$}
\end{eqnarray*}

An \PDL \emph{sentence} is a boolean combination of atomic sentences of the form
$\existsnode{\varphi}$.  An atomic \PDL sentence is evaluated on an order-2
nested word $\nw$.  We have $\nw \models \existsnode{\varphi}$ if there exists a
position $i$ of $\nw$ such that $\nw, i \models \varphi$.

We use abbreviations to include $\true$, $\false$, conjunction, implication,
`$\varphi$ holds after all $\pi$ paths' ($\allpaths{\pi}\varphi$) etc.  We
simply write $\existspath{\pi}$ instead of $\existspath{\pi}\true$ to check the
existence of a $\pi$ path from the current position.  In particular, we can
check the type of a node with $\ispushone={\existspath{\gopopOne}}$,
$\ispopone={\existspath{\gopushOne}}$, and similarly for $\ispushtwo$ and
$\ispoptwo$.  Notice that a collapse node satisfies $\ispoptwo$.  Also,
$\forallnodes{\varphi}=\neg\existsnode{\neg\varphi}$ states that $\varphi$ holds
on all nodes of the \twonw.

\begin{example}\label{ex:matchrel-consitent} 
  We give now path formulas corresponding to the functions $\mypush$ and
  $\myPush$ defined at the end Section~\ref{sec:2cpds}.  We use the
  characterization of Proposition~\ref{prop:mypush}.  Consider first the macro
  $\isfirstpushtwo=\ispushtwo\wedge\neg\existsloop{{\goprev}^+\conc{\gopopTwo}\conc{\gopushTwo}}$
  which identifies uniquely the target of a collapse.
  Then, the deterministic path formulas for $\mypush$ and $\myPush$ are given by
  \begin{align*}
    \pi_{\mypush} &= ( \test{\ispushtwo}\conc{\goprev} 
    + {\gopushOne}\conc{\goprev} 
    + {\gopushTwo}\conc\test{\isfirstpushtwo}\conc{\goprev} )^\ast
    \conc \test{\ispushone}
    \\
    \pi_{\myPush} &= ( \test{\ispushone}\conc{\goprev}
    + \test{\ispopone}\conc{\goprev}
    + {\gopushTwo}\conc\test{\isfirstpushtwo}\conc{\goprev} )^\ast 
    \conc \test{\ispushtwo}
  \end{align*}
  The matching push of an order-1 pop should coincide with the one dictated by
  $\pi_{\mypush}$ starting from the previous node, i.e., before the top symbol
  was popped.  The situation is similar for an order-2 pop.  Hence, the
  following sentence states that $\NestRelOne$ and $\NestRelTwo$ are
  well-nested.
  $$
  \phi_\mathsf{wn} = \forallnodes{\left( 
  ( \ispopone \implies \existsloop{{\goprev}\conc\pi_{\mypush}\conc\gopopOne}) 
  \wedge
  ( \ispushtwo \implies 
  \existsloop{{\gopopTwo}\conc({\goprev}\conc\pi_{\myPush})^+}) 
  \right)} 
  $$
\end{example}

The \emph{satisfiability problem} $\SAT(\PDL)$ asks: Given an \PDL sentence
$\phi$, does there exist a \twonw $\nw$ such that $\nw \models \phi$?  The
\emph{model checking problem} $\MC(\PDL)$ asks, given an \PDL sentence $\phi$
and a \twocpds $\hopds$, whether $\nw \models \phi$ for all \twonw $\nw$ in
$\Lang(\hopds)$.

\begin{theorem}\label{thm:lcpdl-undecidable}
  The problems $\SAT(\PDL)$ and $\MC(\PDL)$ are both undecidable,
  even for \twonw (or \twocpds) without collapse and order-2 stacks of bounded
  height.
\end{theorem}

\begin{proof}

Notice that the \twonws generated by the non-collapsible $\twocpds$ $\hopds_2$ (cf. page~\pageref{ex:twocpdshalfgrid}) embed larger and larger half-grids.  For
instance, the \twonw $\nw_2$ of Figure~\ref{fig:ex2nw} embeds a half-grid of size four.  The lines
are embedded within $\NestRelTwo$: 2, then 6,7, then 11,12,13, and finally  17,18,19,20.
Moving right in the grid amounts to moving right in the \twonw, without crossing
a $\Pop$.  Moving down in the grid (e.g., from 6) amounts to going to the next
order-1 pop $\pop$ (which is 12) attached to the same order-1 push $\push$
(which is 4). 

To prove the undecidability, we encode, in $\PDL$,  the computation of a Turing Machine on the half-grid 
  embedded in the \twonw $\nw_2$.
  First, we write a formula $\grid$ stating that the \twonw is of the correct 
  form.  We use $\last={\existspath{{\gopushOne}}\neg\existspath{\goprev}}$ to
  state that the $\push$ matching the current $\pop$ is the first event of the
  \twonw, hence the top order-1 stack is empty after the current $\pop$. Then, 
  $\grid_1 = \existsnode{(\ispushone \wedge \neg\existspath{\goprev})}$ states 
  that the first event is a $\push$. Next,
  \begin{align*}
    \grid_2 &= \forallnodes{(\ispushone \implies 
    \existspath{{\gonext}\conc\test{\ispushtwo}\conc({\gonext}\conc\test{\ispopone})^+\conc\test{\last}\conc{\gonext}\conc\test{\ispoptwo}}
    )} 
  \end{align*}
  states that the successor of every $\push$ is a $\Push$ followed by a 
  sequence of $\pop$ (the line of the grid) emptying the top order-1 stack, 
  followed by a $\Pop$ which restores the order-1 stack. Finally, 
  $\grid_3 = \forallnodes{(\ispoptwo \implies \neg\existspath{\gonext} \vee
  \existspath{{\gonext}\conc{\gopopOne}})}$ states that a $\Pop$ is either the
  last event of the \twonw, or is followed by a $\push$ starting a new line in 
  the grid. One can check that a \twonw $\nw$ satisfies 
  $\grid=\grid_1\wedge\grid_2\wedge\grid_3$ iff it is of the form of the \twonw 
  depicted above.
  
  We can almost interpret in \PDL the half-grid in a \twonw $\nw$ satisfying
  $\grid$.  Nodes of the grid correspond to $\pop$ events.  Moving right in the
  line of the grid corresponds to the path expression
  ${\gonext}\conc\test{\ispopone}$ and similarly for going left.  Moving down in
  the half-grid (e.g., from 6 to 12 in $\nw_2$), corresponds to going to the
  next-pop-from-same-push in the \twonw.  We do not know whether the next-pop
  relation, denoted $\gonextpop$, can be written as a path expression in \PDL.
  But we have a macro for checking a node formula $\varphi$ at the next-pop:
  \begin{align*}
    \existspath{\gonextpop}\varphi & ::= \existsloop{
    (\test{\ispopone}\conc\gonext)^+\conc\test{\ispoptwo}\conc{\gonext}\conc{\gonext}\conc
    ({\gonext}\conc\test{\ispopone})^+\conc\test{\varphi}\conc{\gopushOne}\conc{\gopopOne}
    }
  \end{align*}
  With this, we can write an \PDL formula to encode the computation of a Turing 
  machine starting from the empty configuration. Consecutive lines of the grid 
  correspond to consecutive configurations. For instance, to check that a 
  transition $(p,a,q,b,{\gonext})$ of the Turing machine is applied at some 
  node, we write the formula $(p \wedge \existspath{\gonext}a) \implies
  \existspath{\gonextpop}(b\wedge\existspath{\gonext}q)$. 
  
  We deduce that $\SAT(\PDL)$ is undecidable.
  Since the \twocpds $\hopds_2$ on page \pageref{ex:twocpdshalfgrid} generates 
  all \twonw satisfying the $\grid$ formula, we deduce that $\MC(\PDL)$ is also
  undecidable.
\end{proof}

\subparagraph{Monadic Second-order Logic} over \twonw, denoted $\MSONW$, extends
the classical MSO over words with two binary predicates $\NestRelOne$ and
$\NestRelTwo$.  A formula $\phi$ can be written using the syntax:
$$ 
\phi := a(x) \mid x <y \mid x \NestRelOne y \mid x \NestRelTwo y \mid \phi \vee \phi \mid \neg \phi \mid \exists x\, \phi \mid \exists X \phi \mid x \in X
$$
where $a \in \Alphabet$, $x, y$ are first-order variables and $X$ is a second-order variable. The semantics is as expected.
As in the case of \PDL, we use common abbreviations.

\begin{example} 
  The binary relation $\nextpop$ which links consecutive pops matching the
  same push can be easily expressed in the first-order fragment as 
  $$
  \phi_{\nextpop} (x, y) = \exists z\, \left(z \NestRelOne x \wedge z 
  \NestRelOne y \wedge \neg \exists z'\, (z \NestRelOne z' \wedge x < z' < y ) 
  \right) \,.
  $$
\end{example}

\begin{example}
  The set of all \twonw can be characterised in \MSONW. It essentially says that
  $<$ is a total order, and that the matching relations are valid.  For the
  latter, we first state that matching relations are compatible with the linear
  order, and that they are disjoint in the following sense: the target of a
  $\NestRelOne$ (resp.\ the source of a $\NestRelTwo$) is not part of another
  matching, and the source of a $\NestRelOne$ (resp.\ the target of a
  $\NestRelTwo$) is not part of $\NestRelTwo$ (resp.\ $\NestRelOne$).
  Finally, to state that $\NestRelOne$ and $\NestRelTwo$ are well-nested,
  we take the idea from Example~\ref{ex:matchrel-consitent}.
\end{example}

\begin{example}\label{ex:mso-language-hopds}
  Given a \twocpds $\hopds$, its language $\Lang(\hopds)$ can be characterised
  by an \MSONW formula.  The formula essentially guesses the transitions taken
  at every position using second-order variables and verifies that this guess
  corresponds to a valid accepting run.  The only difficulty is to check that
  top-tests are satisfied. If the transition guessed for position $x$ contains 
  $\toptest(s)$ then we find position $y$ corresponding to $\mypush(x)$ (expressible by an MSO formula equivalent of the formula $\pi_{\mypush}$) and  
  check that the transition guessed at position $y$ contains $\pushs s$.
\end{example}

The \emph{satisfiability problem} $\SAT(\MSONW)$ asks, given an \MSONW sentence
$\phi$, whether $\nw \models \phi$ for some \twonw $\nw$.  The \emph{model
checking problem} $\MC(\MSONW)$ asks, given an \MSONW sentence $\phi$ and a
\twocpds $\hopds$, whether $\nw \models \phi$ for all \twonw
$\nw\in\Lang(\hopds)$.
Since \PDL can be expressed in \MSONW, we deduce from 
Theorem~\ref{thm:lcpdl-undecidable} that
\begin{theorem}
  The problems $\SAT(\MSONW)$ and $\MC(\MSONW)$ are undecidable.
\end{theorem}

\begin{remark}
  Configuration graphs of \twocpds render MSO undecidable\cite{HagueMOS:2008}.
  For instance, the $\twocpds$ $\hopds_1$ (cf.
  page~\pageref{ex:twocpdshalfgrid}) embeds an infinite half-grid in its
  configuration graph (see, App~\ref{app:grid-configuration-graph}).  Notice
  that \twonw generated by $\hopds_1$ (for instance, $\nw_1$ of
  Figure~\ref{fig:ex2nw}) does not represent the configuration graph, but rather
  some path in it, with extra matching information.
\end{remark}

\newcommand{\marked}[1]{\underline #1}
\newcommand{\transl}[1]{\overline{#1}}
\section{Eliminating collapse}\label{sec:no-collapse}

We can reduce the satisfiability and model checking
problems to variants where there are no collapse operations.  The idea is to
simulate a collapse with a sequence of order-2 pops ($\Pop$).  These pops will
be labelled by a special symbol $\#$ so that we do not confuse it with a normal
order-2 pop.  An internal node is added before such a sequence which indicates
the label of the collapse node (see Figure~\ref{fig:twonw-translation}).
Surprisingly, in \PDL/\MSONW we can express that the number of ${\#}\Pop$ in the
sequence is correct.  We give the details below.

\subparagraph{\twonw to \twonw without collapse.} We expand every collapse node
labelled $a$ by a sequence of the form: $ \marked{a} (\#\Pop)^+$,
with the intention that if we merge all the nodes in this sequence into a single
node labelled by $a\collapse $, we obtain the original $\twonw$ back. Here $\#\Pop$ is a single position, which is labelled $\#$, and is the target of a $\gopopTwo$ edge.  Notice
that the number of $\#\Pop$'s needed for such an encoding is not bounded.  We
will ensure, with the help of \PDL /\MSONW, that the encoding has the precise
number of $\#\Pop$'s needed.

\begin{figure}
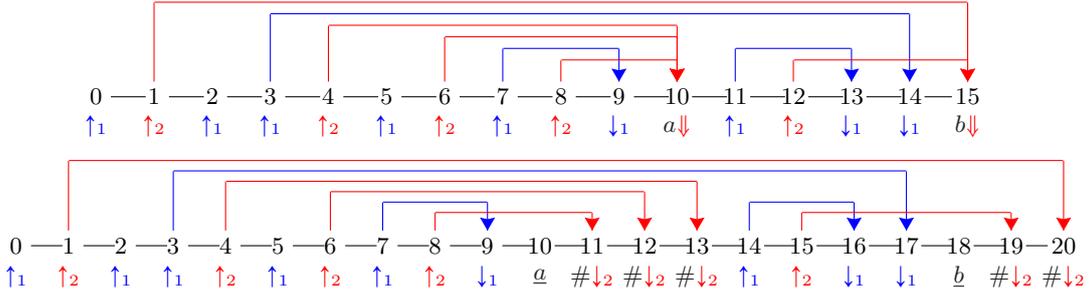

\centering
{
\begin{gpicture}
  \gasset{Nw=4,Nh=4,Nframe=n,AHnb=1,AHLength=2,AHlength=1.5}
  \unitlength=.764705882353mm
  \node(n0)(0,0){\small 0}
  \node(n1)(10,0){\small 1}
  \node(n2)(20,0){\small 2}
  \node(n3)(30,0){\small 3}
  \node(n4)(40,0){\small 4}
  \node(n5)(50,0){\small 5}
  \node(n6)(60,0){\small 6}
  \node(n7)(70,0){\small 7}
  \node(n8)(80,0){\small 8}
  \node(n9)(90,0){\small 9}
  \node(n10)(100,0){\small 10}
  \node(n11)(110,0){\small 11}
  \node(n12)(120,0){\small 12}
  \node(n13)(130,0){\small 13}
  \node(n14)(140,0){\small 14}
  \node(n15)(150,0){\small 15}
  
  \node(e0)(0,-5){\small $\push$}
  \node(e1)(10,-5){\small $\Push$}
  \node(e2)(20,-5){\small $\push$}
  \node(e3)(30,-5){\small $\push$}
  \node(e4)(40,-5){\small $\Push$}
  \node(e5)(50,-5){\small $\push$}
  \node(e6)(60,-5){\small $\Push$}
  \node(e7)(70,-5){\small $\push$}
  \node(e8)(80,-5){\small $\Push$}
  \node(e9)(90,-5){\small $\pop$}
  \node(e10)(100,-5){\small $a\collapse$}
  \node(e11)(110,-5){\small $\push$}
  \node(e12)(120,-5){\small $\Push$}
  \node(e13)(130,-5){\small $\pop$}
  \node(e14)(140,-5){\small $\pop$}
  \node(e15)(150,-5){\small $b\collapse$}

  \drawedge[dash={6 4}0,AHnb=0](n0,n15){}

  \node[Nw=0.1,Nh=0.1](1)(70,8.4){}
  \node[Nw=0.1,Nh=0.1](2)(90,8.4){}
  \drawedge[AHnb=0,linecolor=blue](n7,1){}
  \drawedge[AHnb=0,linecolor=blue](1,2){}
  \drawedge[AHnb=1,linecolor=blue](2,n9){}
  \node[Nw=0.1,Nh=0.1](1)(80,6.4){}
  \node[Nw=0.1,Nh=0.1](2)(100,6.4){}
  \drawedge[AHnb=0,linecolor=red](n8,1){}
  \drawedge[AHnb=0,linecolor=red](1,2){}
  \drawedge[AHnb=1,linecolor=red](2,n10){}
  \node[Nw=0.1,Nh=0.1](1)(60,10.4){}
  \node[Nw=0.1,Nh=0.1](2)(100,10.4){}
  \drawedge[AHnb=0,linecolor=red](n6,1){}
  \drawedge[AHnb=0,linecolor=red](1,2){}
  \drawedge[AHnb=1,linecolor=red](2,n10){}
  \node[Nw=0.1,Nh=0.1](1)(40,12.4){}
  \node[Nw=0.1,Nh=0.1](2)(100,12.4){}
  \drawedge[AHnb=0,linecolor=red](n4,1){}
  \drawedge[AHnb=0,linecolor=red](1,2){}
  \drawedge[AHnb=1,linecolor=red](2,n10){}
  \node[Nw=0.1,Nh=0.1](1)(110,8.4){}
  \node[Nw=0.1,Nh=0.1](2)(130,8.4){}
  \drawedge[AHnb=0,linecolor=blue](n11,1){}
  \drawedge[AHnb=0,linecolor=blue](1,2){}
  \drawedge[AHnb=1,linecolor=blue](2,n13){}
  \node[Nw=0.1,Nh=0.1](1)(30,14.4){}
  \node[Nw=0.1,Nh=0.1](2)(140,14.4){}
  \drawedge[AHnb=0,linecolor=blue](n3,1){}
  \drawedge[AHnb=0,linecolor=blue](1,2){}
  \drawedge[AHnb=1,linecolor=blue](2,n14){}
  \node[Nw=0.1,Nh=0.1](1)(120,6.4){}
  \node[Nw=0.1,Nh=0.1](2)(150,6.4){}
  \drawedge[AHnb=0,linecolor=red](n12,1){}
  \drawedge[AHnb=0,linecolor=red](1,2){}
  \drawedge[AHnb=1,linecolor=red](2,n15){}
  \node[Nw=0.1,Nh=0.1](1)(10,16.4){}
  \node[Nw=0.1,Nh=0.1](2)(150,16.4){}
  \drawedge[AHnb=0,linecolor=red](n1,1){}
  \drawedge[AHnb=0,linecolor=red](1,2){}
  \drawedge[AHnb=1,linecolor=red](2,n15){}
\end{gpicture}
}
\medskip

\begin{gpicture}
  \gasset{Nw=4,Nh=4,Nframe=n,AHnb=1,AHLength=2,AHlength=1.5}
  \unitlength=0.68909090909091mm
  \node(n0)(0,0){\small 0}
  \node(n1)(10,0){\small 1}
  \node(n2)(20,0){\small 2}
  \node(n3)(30,0){\small 3}
  \node(n4)(40,0){\small 4}
  \node(n5)(50,0){\small 5}
  \node(n6)(60,0){\small 6}
  \node(n7)(70,0){\small 7}
  \node(n8)(80,0){\small 8}
  \node(n9)(90,0){\small 9}
  \node(n10)(100,0){\small 10}
  \node(n11)(110,0){\small 11}
  \node(n12)(120,0){\small 12}
  \node(n13)(130,0){\small 13}
  \node(n14)(140,0){\small 14}
  \node(n15)(150,0){\small 15}
  \node(n16)(160,0){\small 16}
  \node(n17)(170,0){\small 17}
  \node(n18)(180,0){\small 18}
  \node(n19)(190,0){\small 19}
  \node(n20)(200,0){\small 20}
  
  \node(e0)(0,-6){\small $\push$}
  \node(e1)(10,-6){\small $\Push$}
  \node(e2)(20,-6){\small $\push$}
  \node(e3)(30,-6){\small $\push$}
  \node(e4)(40,-6){\small $\Push$}
  \node(e5)(50,-6){\small $\push$}
  \node(e6)(60,-6){\small $\Push$}
  \node(e7)(70,-6){\small $\push$}
  \node(e8)(80,-6){\small $\Push$}
  \node(e9)(90,-6){\small $\pop$}
  \node(e10)(100,-6){\small $\marked a$}
  \node(e11)(110,-6){\small $ \#\Pop$}
  \node(e12)(120,-6){\small $ \#\Pop$}
  \node(e13)(130,-6){\small $ \#\Pop$}
  \node(e14)(140,-6){\small $\push$}
  \node(e15)(150,-6){\small $\Push$}
  \node(e16)(160,-6){\small $\pop$}
  \node(e17)(170,-6){\small $\pop$}
  \node(e18)(180,-6){\small $\marked b$}
  \node(e19)(190,-6){\small $ \#\Pop$}
  \node(e20)(200,-6){\small $ \#\Pop$}

  \drawedge[dash={6.0 4.0}0,AHnb=0](n0,n20){}

  \node[Nw=0.1,Nh=0.1](1)(70,8.4){}
  \node[Nw=0.1,Nh=0.1](2)(90,8.4){}
  \drawedge[AHnb=0,linecolor=blue](n7,1){}
  \drawedge[AHnb=0,linecolor=blue](1,2){}
  \drawedge[AHnb=1,linecolor=blue](2,n9){}
  \node[Nw=0.1,Nh=0.1](1)(80,6.4){}
  \node[Nw=0.1,Nh=0.1](2)(110,6.4){}
  \drawedge[AHnb=0,linecolor=red](n8,1){}
  \drawedge[AHnb=0,linecolor=red](1,2){}
  \drawedge[AHnb=1,linecolor=red](2,n11){}
  \node[Nw=0.1,Nh=0.1](1)(60,10.4){}
  \node[Nw=0.1,Nh=0.1](2)(120,10.4){}
  \drawedge[AHnb=0,linecolor=red](n6,1){}
  \drawedge[AHnb=0,linecolor=red](1,2){}
  \drawedge[AHnb=1,linecolor=red](2,n12){}
  \node[Nw=0.1,Nh=0.1](1)(40,12.4){}
  \node[Nw=0.1,Nh=0.1](2)(130,12.4){}
  \drawedge[AHnb=0,linecolor=red](n4,1){}
  \drawedge[AHnb=0,linecolor=red](1,2){}
  \drawedge[AHnb=1,linecolor=red](2,n13){}
  \node[Nw=0.1,Nh=0.1](1)(140,8.4){}
  \node[Nw=0.1,Nh=0.1](2)(160,8.4){}
  \drawedge[AHnb=0,linecolor=blue](n14,1){}
  \drawedge[AHnb=0,linecolor=blue](1,2){}
  \drawedge[AHnb=1,linecolor=blue](2,n16){}
  \node[Nw=0.1,Nh=0.1](1)(30,14.4){}
  \node[Nw=0.1,Nh=0.1](2)(170,14.4){}
  \drawedge[AHnb=0,linecolor=blue](n3,1){}
  \drawedge[AHnb=0,linecolor=blue](1,2){}
  \drawedge[AHnb=1,linecolor=blue](2,n17){}
  \node[Nw=0.1,Nh=0.1](1)(150,6.4){}
  \node[Nw=0.1,Nh=0.1](2)(190,6.4){}
  \drawedge[AHnb=0,linecolor=red](n15,1){}
  \drawedge[AHnb=0,linecolor=red](1,2){}
  \drawedge[AHnb=1,linecolor=red](2,n19){}
  \node[Nw=0.1,Nh=0.1](1)(10,16.4){}
  \node[Nw=0.1,Nh=0.1](2)(200,16.4){}
  \drawedge[AHnb=0,linecolor=red](n1,1){}
  \drawedge[AHnb=0,linecolor=red](1,2){}
  \drawedge[AHnb=1,linecolor=red](2,n20){}
\end{gpicture}
\caption{A \twonw (top) and its encoding in \twonw without collapse (below).  We
show the labels only on the nodes of interest.  }\label{fig:twonw-translation}
\end{figure}

\subparagraph{\twocpds\ to \twocpds\ without collapse.} Given an \twocpds\
$\hopds$, we construct a new \twocpds\ $\hopds'$ where, for each collapse
transition $t = (\state, a, \collapse, \state')$ there is an extra state
$\state_t$.  Further, instead of the transition $t$ we have the following three
transitions: $(\state, \marked a, \nop, \state_t)$, $(\state_t, \#, \Pop,
\state_t)$, and $(\state_t, \#, \Pop, \state')$.  Notice that for every $\twonw$
in the language of $\hopds$, its encoding without collapse will be in the
language of $\hopds'$.  However, the language of $\hopds'$ may contain spurious
runs where the $(\state_t, \#, \Pop, \state_t)$ is iterated an incorrect number
of times.  These spurious runs will be discarded by adding a precondition to the
specification.

\subparagraph{\PDL to \PDL without collapse.} 
We identify a collapse node by the first node in a block of the form 
$\marked{a}(\#\Pop)^+$.  We
call the positions labelled by symbols other than $\#$ \emph{representative
positions}.  Intuitively, we will be evaluating node formulas only at
representative positions, and path formulas connect a pair of representative
positions.  Checking whether the current node is labelled $a$, would now amount
to checking whether the current node is labelled by $a$ or $\marked a$.
Further, in path formulas, moving to right ($\gonext$) would mean going to the
next representative position in the right.  Similarly for $\goprev$.  The path
formula $\gopopTwo$ would correspond to taking the $\gopopTwo$ edge and moving
left until a representative position is reached.  
Notice that $\gopushTwo$ at a collapse node can non-deterministically choose any
$\Push$ matched to it.  Hence we express this as $(\gonext \conc \test{\#})^\ast
\conc \gopushTwo$.  Notice that this formula also handles the case when the
current node is not of the form $\marked{a}$.  The other modalities remain
unchanged.  This translation can be done in linear time and the size of the
resulting formula is linear in the size of the original formula. 
Translation of a \PDL node formula $\varphi$ and path formula $\pi$ to \PDL
without collapse is given below.  The translation is denoted $\transl{\varphi}$
and $\transl{\pi}$ respectively.

\medskip\noindent\hfil
$\begin{array}{rcl}
\transl{a} &\equiv & a \vee \marked{a}\\
\transl{\varphi_1 \vee \varphi_2} &\equiv & \transl{\varphi_1} \vee \transl{\varphi_2}\\
\transl{\neg \varphi} & \equiv & \neg \transl{\varphi}\\
\transl{\existspath{\pi}\varphi} &\equiv & \existspath{\transl\pi} \transl\varphi\\
\transl{\existsloop{\pi}} &\equiv & \existsloop{\transl\pi}
\end{array}$
\hfil
$\begin{array}{rcl}
\transl{\test \varphi } &\equiv & \test{\transl \varphi}\\
\transl{\pi_1 \conc \pi_2} &\equiv & \transl{\pi_1} \conc \transl{\pi_2}\\
\transl{\pi_1 + \pi_2} &\equiv & \transl{\pi_1} + \transl{\pi_2}\\
\transl{{\pi}^\ast} &\equiv & {\transl\pi}^\ast\\
\transl{\gopopOne} &\equiv & \gopopOne
\end{array}$
\hfil
$\begin{array}{rcl}
\transl{\gopushOne} &\equiv & \gopushOne\\
\transl{\gonext } &\equiv & {\gonext}\conc(\test{\#}\conc{\gonext})^\ast\conc\test{\neg \#}\\
\transl{\goprev} &\equiv & {\goprev}\conc(\test{\#}\conc{\goprev})^\ast\conc\test{\neg \#}\\
\transl{\gopushTwo} &\equiv & ({\gonext} \conc \test{\#})^\ast \conc {\gopushTwo}\\
\transl{\gopopTwo} &\equiv & {\gopopTwo} \conc (\test{\#} \conc {\goprev})^\ast \conc \test{\neg \#}
\end{array}$

\subparagraph{\MSONW to \MSONW without collapse.} Translation of \MSONW is similar in spirit to that of \PDL. Every atomic formula (binary relations and unary predicates) is translated as done in the case of \PDL. The rest is then a standard relativisation to the representative positions. Overloading notations, we denote the translation of an \MSONW formula $\varphi$ by $\transl \varphi$.

\newcommand{\mypushone}{\pi_{\mypush}}
\newcommand{\mypushtwo}{\pi_{\myPush}}

\subparagraph{Identifying valid encodings with \PDL/\MSONW. } First we need to
express that the number of $\#\Pop$ is correct.  Towards this, we will state
that, the matching push of the last pop in a $\marked{a}(\#\Pop)^+$ block indeed
corresponds to the push of the stack in which the topmost stack symbol was
pushed.  We explain this below.  From a representative position labelled
$\marked a$ we can move to the $\push$ position $x$ where the topmost stack
symbol of the topmost stack was pushed, by taking a $\mypushone$ path.  The
$\Push$ position $y$ which pushed the stack which contains the element pushed at
$x$ can be reached by a $\mypushtwo$ path from $x$.  We then say that the
matching pop of $y$ is indeed the last $\#\Pop$ labelled node in the sequence,
and that it indeed belongs to the very sequence of $\marked a$ we started with.
We can state this in \PDL with:
$$
\phi_{\mathsf{valid}} =
\forallnodes \bigwedge_{a\in\Alphabet} \left( \marked a \implies
\existsloop {{\mypushone}\conc{\mypushtwo}
\conc {\gopopTwo}\conc \test{\neg \existspath{\gonext}\#} \conc
(\test{\#}\conc{\goprev})^+ } \right) \,.
$$

\subparagraph{Satisfiability and model checking problems.} 
The satisfiability problem of \PDL/\MSONW formula $\phi$ over
\twonw with collapse reduces to the satisfiability problem of
$\phi_{\mathsf{valid}} \wedge \transl{\phi}$ over \twonw without collapse.
The model checking problem of \twocpds $\hopds$ against a \PDL/\MSONW formula
$\phi$ reduces to the model checking problem of $\hopds'$ against
$\phi_{\mathsf{valid}} \implies \transl{\phi}$.

\begin{remark}
  In \cite{AMO05} it is shown that for every $\twocpds$ there exists an order-2
  pushdown system without collapse generating the same word language (without
  nesting relations).  Our result of this section shows that, model checking and
  satisfiability checking, even in the presence of nesting relations, can be
  reduced to the setting without collapse.
\end{remark}

\section{Bounded-pop \twonw}\label{sec:under-approximation}

In this section we will define an under-approximation of \twonws which regains decidability of the verification problems discussed in
Section~\ref{sec:logic}.

\subparagraph{Bounded-pop \twonws} are \twonws in which a pushed symbol may be
popped at most a bounded number of times.  This does not limit the number of
times an order-1 stack may be copied, nor the height of any order-1 or order-2
stack.
Our
restriction amounts to bounding the number of $\NestRelOne$ partners that a push
may have.  Let $\bound$ denote this bound for the rest of the paper.  The class
of \twonw in which every order-1 push has at most $\bound$ many matching pops is
called \emph{$\bound$-pop-bounded order-2 nested words}.  It is denoted $\pbnw$.

\subparagraph{Bounded-pop model checking} The under-approximate satisfiability
problem and model checking problems are defined as expected.  The problem
$\SAT(\PDL,\beta)$ (resp.\ $\SAT(\MSONW,\beta)$) asks, given an $\PDL$ (resp.\
\MSONW) sentence $\phi$ and a natural number $\bound$, whether $\nw \models
\phi$ for some \twonw $\nw\in\pbnw$.  
The problem $\MC(\PDL,\beta)$ (resp.\ $\MC(\MSONW,\beta)$) asks, given an $\PDL$
(resp.\ \MSONW) sentence $\phi$, a \twocpds $\hopds$ and a natural number
$\bound$, whether that $\nw\models\phi$ for all \twonw $\nw \in \Lang(\hopds)
\cap \pbnw$.

\begin{theorem}\label{thm:main}
  The problems $\SAT(\PDL,\beta)$ and $\MC(\PDL,\beta)$ are \textsc{ExpTime}-Complete. 
  The problems $\SAT(\MSONW,\beta)$ and $\MC(\MSONW,\beta)$ are decidable.
\end{theorem}

\twocpds can simulate nested-word automata (NWA) \cite{AlurM09} by not
using any order-1 stack operations.  $\Push$ and $\Pop$ will play the role of
push and pop of NWA. Satisfiability of PDL over nested words, and model checking
of PDL against NWA are known to be \textsc{ExpTime}-Complete
\cite{BCGZ-jal14,BCGZ-mfcs11}.  The \textsc{ExpTime}-hardness of
$\SAT(\PDL,\beta)$ and $\MC(\PDL,\beta)$ follows.
The decision procedures establishing Theorem~\ref{thm:main} are given in the
next section. Thanks to Section~\ref{sec:no-collapse} we will restrict our 
attention to $\twonw$ without collapse.

\section{Split-width and decision procedures}\label{sec:split-width}

In all of this section, by \twonws we mean order-2 nested words \emph{without
collapse}.  We lift the notion of split-decomposition and split-width to \twonws
and show that words in $\pbnw$ have split-width bounded by $2\beta+2$.  Then we
show that nested words in $\pbnw$ can be interpreted in binary trees, which is
the core of our decision procedures.

\subparagraph{A Split-\twonw} 
is a \twonw in which the underlying word has been split in several factors.
Formally, a split-\twonw is a tuple
$\snw=\tuple{u_1,\ldots,u_m,\NestRelOne,\NestRelTwo}$ such that
$\nw=\tuple{u_1\cdots u_m,\NestRelOne,\NestRelTwo}$ is a \twonw.
The number $m$ of factors in a split-\twonw is called it's width.
A \twonw is a split-\twonw of width one.

A split-\twonw can be seen as a labelled graph whose vertices are the
positions of the underlying word (concatenation of the factors) and we have order-1 edges $\NestRelOne$, 
order-2 edges $\NestRelTwo$ and successor edges $\gonext$ between consecutive 
positions \emph{within} a factor. 
We say that a split-\twonw is connected if the underlying graph is \emph{connected}. 
If a split-\twonw is not connected, then its connected components form a
partition of its factors.

\smallskip
\noindent
\begin{minipage}{.7\textwidth}
\begin{example} 
  Consider the split-\twonw on the right.
 It has two connected components, and its width is five.
\end{example}
\end{minipage}
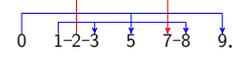
\begin{minipage}{.3\textwidth}
 \hfill\includegraphics[scale=.6,page=6]{tikz-pics}.
\end{minipage}

\subparagraph{The split game} 
is a two-player turn based game $\mathcal{G} = \tuple{V = V_\forall \uplus V_\exists, E}$ where 
\Eve's positions $V_\exists$ consists of (connected or not)
split-\twonws, and
\Adam's positions $V_\forall$ consists of non connected split-\twonws.
The edges $E$ of $\mathcal{G}$ reflect the moves of the players.
\Eve's moves consist in removing some successor 
edges in the graph, i.e., splitting some factors, so that the resulting graph 
is not connected.
\Adam's moves amounts to choosing a connected component.
A split-\twonw is atomic if it is connected and all
its factors are singletons.  An atomic split-\twonw contains either a single
internal event, or, a single push with all its corresponding pops. 
A play on a split-\twonw $\snw$ is path in $\mathcal{G}$ starting from $\snw$ to
an atomic split-\twonw.  The cost of the play is the maximum width of any
split-\twonw encountered in the path.  \Eve's objective is to minimize the cost
and \Adam's objective is to maximize the cost.

A strategy for \Eve from a split-\twonw $\snw$ can be described with a
\emph{split-tree} $T$ which is a binary tree labelled with split-\twonw 
satisfying:

\noindent
\begin{minipage}{.6\textwidth}
  \begin{enumerate}[nosep]
    \item The root is labelled by $\snw$ and
    leaves are labelled by atomic split-\twonw.

    \item \Eve's move: Each unary node is labelled with some 
    split-\twonw $\snw$ and its child is labelled with $\overline{\nw'}$ obtained 
    by splitting some factors of $\snw$.

    \item \Adam's move: Each binary node is labelled with some non connected
    split-\twonw $\snw = \snw_1 \uplus \snw_2$ where 
    $\snw_1$, $\snw_2$ are the
    labels of its children.  Note that 
    $\width(\snw)=\width(\snw_1)+\width(\snw_2)$.
  \end{enumerate}
  The \emph{width} of a split-tree $T$, denoted $\width(T)$, is the maximum width
  of the split-\twonws labelling $T$.
  The \emph{split-width} of a split-\twonw $\snw$ is the minimal width of all split-trees for $\snw$. 
  A split-tree is depicted
  above. 
  The width of the split-\twonw labelling binary nodes are five.  Hence the
  split-width of the split-\twonw labelling the root is at most five.
\end{minipage}
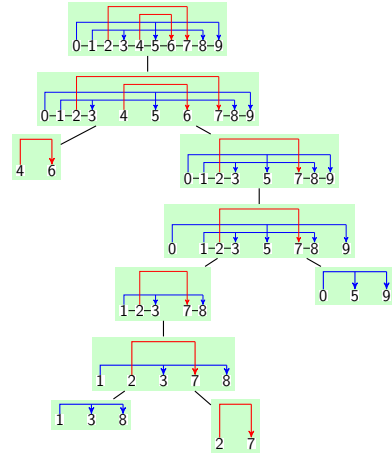
\begin{minipage}{.4\textwidth}
  \hfill\includegraphics[scale=.52,page=3]{tikz-pics}
\end{minipage}

\begin{theorem}\label{thm:swofpopbounded}
  Nested words in $\pbnw$ have split-width bounded by $k=2\bound+2$.
\end{theorem}

\newcommand{\close}[1]{\textsf{Close}(#1)}
\newcommand{\cp}{\textsf{ctxt-pop}}
\newcommand{\cs}{\textsf{ctxt-suffix}}

\begin{proof}
First, we say that a split-word $\overline{\nw} 
=\tuple{u_0,u_1,\ldots,u_m,\NestRelOne,\NestRelTwo} $ 
with $m \le \bound$ is
\emph{good} if $\close{\overline{\nw}}$ = \raisebox{-3pt}{\tikzclose}
is a valid \twonw 
(nesting edges between the factors are not depicted).
Notice that any \twonw $\nw=\close{\nw}$ is vacuously good.

Our strategy is to decompose good split-words into atomic split-words or smaller
good split-words so that we can proceed inductively.  Good split-words have
width at most $\bound+1$.  On decomposing a good split-word to obtain smaller good
split-words, we may temporarily generate (not necessarily good) split-words of
higher width, but we never exceed $k = 2\bound +2$.

\newcommand{\tikzcaseone}{\includegraphics[scale=1,page=16]{tikz-pics}}
\newcommand{\tikzcaseoneone}{\includegraphics[scale=1,page=17]{tikz-pics}}
\newcommand{\tikzcaseonetwo}{\includegraphics[scale=1,page=18]{tikz-pics}}
\newcommand{\tikzcaseonethree}{\includegraphics[scale=1,page=19]{tikz-pics}}
\newcommand{\tikzcaseonefour}{\includegraphics[scale=1,page=20]{tikz-pics}}

Consider any good split-word $\overline{\nw} $.  We have two cases: either it
begins with a $\push$, or it begins with a $\Push$.

If $\overline{\nw} $ begins with a $\Push$: $\overline{\nw} = $
\raisebox{-3pt}{\tikzcaseone}.
We split factors $u_0$ and $u_i$ to get $\overline{\nw'}$ of width at most $m+4$. \\
$\overline{\nw'}=$
\raisebox{-3pt}{\tikzcaseoneone}
Note that, there cannot be any
$\NestRelOne$ edges from $u'_0,\ldots,u^1_i$ to $u^2_i,\ldots,u_m$ since the
duplicated stack is lost at the $\Pop$.  Further, there cannot be any
$\NestRelTwo$ edges from $u'_0,\ldots,u^1_i$ to $u^2_i,\ldots,u_m$ since
$\NestRelTwo$ is well-nested.  Hence, $\overline{\nw'}$ can be divided into
atomic 
\raisebox{-3pt}{\tikzcaseonefour} 
and split-words $\overline{\nw_1} = $
\raisebox{-3pt}{\tikzcaseonetwo}
and $\overline{\nw_2} = $ 
\raisebox{-3pt}{\tikzcaseonethree}.
Since
$\overline{\nw}$ is good, we can prove that $\overline{\nw_1} $ and
$\overline{\nw_2} $ are also good.
 
\medskip%
Before moving to the more involved case where $\overline{\nw} $ begins with a
$\push$, we first see a couple of properties of \twonw which become handy.
Examples are given in Appendix~\ref{sec:context-pop}.

The \emph{context-pop} of an order-1 pop at position $x$, denoted $\cp(x)$, is
the position of the order-2 pop of the innermost $\NestRelTwo$ edge that
encloses $x$.  If $z \NestRelTwo y$ and $z < x < y$ and there is no $z'
\NestRelTwo y'$ with $z < z' < x < y' < y$, then $\cp(x) = y$.  If $x$ does not
have a context-pop, then it is \emph{top-level}.

Consider any \twonw $\tuple{a_1a_2w, \NestRelOne, \NestRelTwo}$ beginning with
two \mbox{order-1} pushes with corresponding pops at positions $y_1^1, \ldots,
y_m^1$ and $y_1^2, \ldots, y_n^2$ respectively.  Then for each $y_i^1$, there is
a $y_j^2$ with $y_j^2 < y_i^1$ such that either $\cp(y_i^1) \le \cp(y_j^2)$ or
$y_j^2$ is top-level.
We call this property \emph{existence of covering pop} for later reference.

Consider any \twonw $\tuple{aw, \NestRelOne, \NestRelTwo}$ beginning with an
order-1 push with corresponding pops at positions $y_1, \ldots , y_m$.  The
\emph{context-suffix} of $y_i$, denoted by $\cs(y_i)$ is the factor of $w$ 
strictly between
$y_i$ and $\cp(y_i)$, both $y_i$ and $\cp(y_i)$ not included.  If $y_i$ is
top-level, then $\cs(y_i)$ is the biggest suffix of $w$ not including $y_i$.
We can prove that, for each $i$, $\cs(y_i)$ is not connected to
the remaining of $w$ via $\NestRelOne$ or $\NestRelTwo$ edges.  The notion of
$\cs(y_i)$ may be lifted to split-words $\overline{\nw} = \tuple{au_0, u_1, u_2,
\ldots, u_n, \NestRelOne, \NestRelTwo}$ also.  In this case, $\cs(y_i)$ may
contain several factors.  Still $\cs(y_i)$ is not connected to the remaining
factors.  Moreover, if $\overline{\nw}$ is good, so is $\cs(y_i)$ .

\medskip%
Now we are ready to describe the decomposition for the second case where the
good split-word $\overline{\nw} $ begins with a $\push$.  Let
$\overline{\nw} = \tuple{au_0, u_1, u_2, \ldots u_n, \NestRelOne, \NestRelTwo}$
beginning with an order-1 push with corresponding pops at positions $y_1, \ldots
, y_m$.  Note that $n, m \le \bound$.  We proceed as follows.  We split at most
two factors of $\overline{\nw}=\overline{\nw_m}$ to get $\overline{\nw_m'}$ in
which $\cs(y_m)$ is a (collection of) factor(s).  Recall that $\cs(y_m)$ is not
connected to other factors.  Hence we divide the split-word $\overline{\nw_m'}$
to get $\cs(y_m)$ as a split-word and the remaining as another split-word
$\overline{\nw_{m-1}}$.  Note that $\cs(y_m)$ is good, so we can inductively
decompose it.  We proceed with $\overline{\nw_{m-1}}$ (which needs not be good).  We
split at most two factors of $\overline{\nw_{m-1}}$ to get $\overline{\nw_{m-1}'}$
in which $\cs(y_{m-1})$ is a (collection of) factor(s).  We divide
$\overline{\nw_{m-1}'}$ to get $\cs(y_{m-1})$ which is good, and
$\overline{\nw_{m-2}}$.  We proceed similarly on $\overline{\nw_{m-2}}$ until we
get $\overline{\nw_0}$.

Note that the width of $\overline{\nw_i}$ (resp.\ $\overline{\nw_i'}$) is at
most $n+1+m-i$ (resp.\ $n+1+m-i+2$).  Hence the width of this stretch of
decomposition is bounded by $n + m + 2$.  Since $n, m \le \bound$, the bound $k
= 2 \bound +2$ of split-width is not exceeded.

Now we argue that the width of $\overline{\nw_1}$ is at most $m+1$.  This is where
we need the invariant of being good.  Consider $\nw =
\close{\overline{\nw}}$, which is a \twonw since $\overline{\nw}$ is good.
Now, $\nw$ starts with two order-1 pushes. Using the existence of covering pop
property, we deduce that every split/hole in $\overline{\nw}$ must
belong to some $\cs(y_i)$.  Hence, all splits from $\overline{\nw}$ are removed
in $\overline{\nw_0}$, and only $m$ splits corresponding to the removed
$\cs(y_i)$ persist.  

We proceed with $\overline{\nw_0}$.  We make at most $m+1$ new splits to get
$\overline{\nw'}$ in which the first push and its pops at positions
$y_1,\ldots,y_m$ are singleton factors.  The width of $\overline{\nw'}$ is at
most $2m +2\leq 2 \bound +2$.  Then we divide $\overline{\nw'}$ to get atomic
split-word consisting of the first push and its pops, and another split-word
$\overline{\nw''}$.  The width of $\overline{\nw''}$ is at most $m+1\le\bound+1$.
Further $\overline{\nw''}$ is good. 
By induction $\overline{\nw''}$ can also be decomposed.
\end{proof}

In Appendix~\ref{sec:tree-decision}, we show that \twonws of split-width at most $k$ have special
tree-width (\STW) at most $2k$.  We deduce that \twonws of bounded split-width can
be interpreted in special tree terms (\STTs), which are binary trees denoting
graphs of bounded STW.
Special tree-width and special tree terms were introduced by Courcelle in 
\cite{Courcelle10}.

A crucial step towards our decision procedures is then to construct a tree
automaton \Aswk which accepts special tree terms denoting graphs that are \twonw
of split-width at most $k$.  The main difficulty is to make sure that the edge
relations $\NestRelOne$ and $\NestRelTwo$ of the graph are well nested.  To
achieve this with a tree automaton of size $2^{\poly(k)}$, we use the
characterization given by the \PDL formula $\phi_\mathsf{wn}$ of
Example~\ref{ex:matchrel-consitent}.  Similarly, we can construct a tree
automaton \Apopb of size $2^{\poly(\bound)}$ accepting \STTs denoting nested
words in $\pbnw$.

Next, we show that for each \twocpds $\hopds$ we can construct a tree automaton
\Ahopds of size $2^{\poly(\bound,\sizeof{\hopds})}$ accepting \STTs denoting
nested words in $\pbnw$ which are accepted by $\hopds$.  We deduce that
non-emptiness checking of \twocpds with respect to $\pbnw$ is in
\textsc{ExpTime}.

Finally, for each \PDL formula $\phi$ we can construct a tree automaton \Aphi of
size $2^{\poly(\bound,\sizeof{\phi})}$ accepting \STTs denoting nested words in $\pbnw$
which satisfy $\phi$. We deduce that $\SAT(\PDL,\beta)$ and $\MC(\PDL,\beta)$ 
can be solved in \textsc{ExpTime}.

Similarly, for each \MSONW formula $\phi$ we can construct a tree automaton
\Aphi accepting \STTs denoting nested words in $\pbnw$ which satisfy $\phi$.  We
deduce that $\SAT(\MSONW,\beta)$ and $\MC(\MSONW,\beta)$ are decidable.

\section{Related work}\label{sec:comparison}
In \cite{Broadbent12}, Broadbent studies nested structures of order-2 HOPDS. A
suffix rewrite system that rewrites nested words is used to capture the graph of
$\epsilon$-closure of an order-2 HOPDS. The objective of the paper as well as
the use of nested words is different from ours.

\paragraph{Nested trees.} {\twonw}s have close relation with nested trees. A nested tree \cite{AlurCM06,AlurCM11}
is a tree with an additional binary relation such that every branch forms a
well-nested word\cite{AlurM09}.  It provides a ``visible''
representation of the branching behaviour of a pushdown system.  

Every finite nested tree can be embedded inside a \twonw without collapse.  In
our encoding, order-1 matching relation captures the nesting relation in the
nested tree, and order-2 matching relation captures the branching structure.
See the example below:\\
\gusepicture{gpic:nested-tree}\hfill\gusepicture{gpic:tree-to-nw}

(Encoding of) every left sub-tree is enclosed within a 
$\bullet\NestRelTwo\bullet$ pair, whose source and target nodes are not part of the original nested-tree.  There is a bijective correspondence between the other nodes in the \twonw and nodes in the nested-tree, and the edges in the nested-tree can be easily MSO-interpreted in the \twonw.  This implies that MSO over \twonw is undecidable, since it
is undecidable over nested-trees \cite{AlurCM11}.

Conversely, every \twonw can be MSO interpreted in a nested-tree.  We may assume
that $\twonw$ is collapse-free (cf.  Section~\ref{sec:no-collapse}).  The $\Pop$
nodes are linked in the tree as the right child of the matching push, and the
other nodes are linked as the left child of its predecessor.

\medskip\noindent\hfil
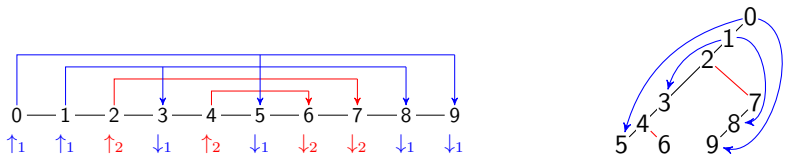
\includegraphics[scale=.8,page=9]{tikz-pics}
\hfil
\includegraphics[scale=1,page=10]{tikz-pics}

On nested trees, $\mu$-calculus was shown to be decidable in \cite{AlurCM06}. First-order logic over nested trees, when the
signature does not contain the order $<$ (but only the successor $\gonext$) is shown decidable in 
\cite{kartzow13}, but MSO over nested trees is undecidable \cite{AlurCM06}.  
Our under-approximation of bounded-pop which regains decidability
corresponds to bounding the number of matching pops that a tree-node can have,
which in turn bounds the degree of the nodes in the tree.

\section{Conclusion}
In this paper we study the linear behaviour of a \twocpds by giving extra structure to words. The specification formalisms can make use of this structure to describe properties of the system. This added structure comes with the cost of undecidable verification problems. We identify an under approximation that  regains decidability for verification problems. Our decision procedure makes use of the split-width technique.

This work is a first step towards further questions that must be investigated.
One direction would be to identify other under-approximations which are
orthogonal / more lenient than bounded-pop for decidability.  Whether similar
results can be obtained for order-n CPDS is also another interesting future
work.  The language theory of CPDS where the language consists of nested-word
like structures is another topic of interest.

\clearpage


\clearpage

\appendix

\section{Characterisation of the matching relations}\label{app:characterisation}

\rstpropmypush*

\begin{proof}
  Notice that \eqref{item:P3} and \eqref{item:P4} follow directely from 
  \eqref{item:P1} and \eqref{item:P2}.
  We prove \eqref{item:P1} and \eqref{item:P2} simultaneously by induction on
  $j$ and by a case splitting depending on $\op_j$.
  \begin{itemize}
    \item Assume $\op_j=\push$. 
    
    \eqref{item:P1} We have, $\mypush(j)=j$. Moreover, 
    $\op_i\cdots\op_j\in L_1$ iff $i=j$.
    
    \eqref{item:P2} We have $\op_i\cdots\op_j\in L_2$ 
    iff $j>i$ and $\op_i\cdots\op_{j-1}\in L_2$ 
    iff $j>i$ and (by induction) $\myPush(j-1)=i$ 
    iff $\myPush(j)=i$.
  
    \item Assume $\op_j=\Push$. 
    
    \eqref{item:P1} We have $\op_i\cdots\op_j\in L_1$ 
    iff $j>i$ and $\op_i\cdots\op_{j-1}\in L_1$ 
    iff $j>i$ and (by induction) $\mypush(j-1)=i$ 
    iff $\mypush(j)=i$.
    
    \eqref{item:P2} We have, $\myPush(j)=j$. Moreover, 
    $\op_i\cdots\op_j\in L_2$ iff $i=j$.

    \item Assume $\op_j=\pop$. 
  
    \eqref{item:P1} We have $\op_i\cdots\op_j\in L_1$ 
    iff $\op_i\cdots\op_{k-1}\in L_1$ and $\op_k\cdots\op_{j-1}\in L_1$ for some $i<k<j$ 
    iff (by induction) $\mypush(k-1)=i$ and $\mypush(j-1)=k$ for some $i<k<j$ 
    iff $i=\mypush(\mypush(j-1)-1)$ 
    iff $\mypush(j)=i$.

    Let us explain the last equivalence.  The symbol which is popped by
    $\op_j=\pop$ was the top symbol at $j-1$, which was pushed at
    $\mypush(j-1)=k$.  We deduce that the top symbol after $\op_j=\pop$ is the
    top symbol before $\op_k=\push$.  Therefore, $\mypush(j)=\mypush(k-1)$.
    
    \eqref{item:P2}  We have $\op_i\cdots\op_j\in L_2$ 
    iff $j>i$ and $\op_i\cdots\op_{j-1}\in L_2$ 
    iff $j>i$ and (by induction) $\myPush(j-1)=i$ 
    iff $\myPush(j)=i$.

    \item Assume $\op_j=\Pop$.
  
    \eqref{item:P1} We have $\op_i\cdots\op_j\in L_1$ 
    iff $\op_i\cdots\op_{k-1}\in L_1$ and $\op_k\cdots\op_{j-1}\in L_2$ for some $i<k<j$ 
    iff (by induction) $\mypush(k-1)=i$ and $\myPush(j-1)=k$ for some $i<k<j$ 
    iff $i=\mypush(\myPush(j-1)-1)$
    iff $\mypush(j)=i$.

    Let us explain the last equivalence.  The order-1 stack which is popped by
    $\op_j=\Pop$ was the top order-1 stack at $j-1$, which was pushed at
    $\myPush(j-1)=k$.  We deduce that the top order-1 stack after $\op_j=\Pop$
    is the top order-1 stack before $\op_k=\Push$.  Hence, the top symbol after
    $\op_j=\pop$ is the top symbol before $\op_k=\push$.  Therefore,
    $\mypush(j)=\mypush(k-1)$.
    
    \eqref{item:P2} The proof is obtained mutatis mutandis from the proof of
    \eqref{item:P1} above. 

    \item Assume $\op_j=\collapse$.

    \eqref{item:P1} We have $\op_i\cdots\op_j\in L_1$ iff
    $\op_i\cdots\op_{k-1}\in L_1$ and $\op_k\cdots\op_{\ell-1}\in L_2$ and
    $\op_\ell\cdots\op_{j-1}\in L_1$ for some $i<k<\ell<j$ iff (by induction)
    $\mypush(k-1)=i$ and $\myPush(\ell-1)=k$ and $\mypush(j-1)=\ell$ for some 
    $i<k<\ell<j$ iff
    $i=\mypush(\myPush(\mypush(j-1)-1)-1)$ iff $i=\mypush(j)$.

    Let us explain the last equivalence.  
    The collapse operation depends on the top symbol before $\op_j=\collapse$.
    This symbol was pushed at $\mypush(j-1)=\ell$. The collapse link which was 
    created by $\op_\ell=\push$ points to the order-1 stack just below the top 
    order-1 stack before (or after) $\op_\ell$. This top order-1 stack was 
    pushed at $k=\myPush(\ell-1)=\myPush(\ell)$. Now, the order-2 stack after 
    $\op_j=\collapse$ is exactely the order-2 stack before $\op_k=\Push$.
    Therefore, $\mypush(j)=\mypush(k-1)$.
    
    \eqref{item:P2} The proof is obtained mutatis mutandis from the proof of
    \eqref{item:P1} above.
    \qedhere
\end{itemize}
  
\end{proof}

\section{Infinite half grid in the configuration graph of $\hopds_1$  }\label{app:grid-configuration-graph}

\hfil\includegraphics[scale=1,page=23]{tikz-pics}

\section{Example showing context-pop and context-suffix}
\label{sec:context-pop}

Consider the \twonw below. \\[1ex]
\noindent
\includegraphics[scale=1,page=11]{tikz-pics}

\medskip
\noindent We have, 

\begin{tabular}{p{4.6cm} l}
$\cp(3)=12$ & $\cp(5)=\cp(7)=8$ \\
 $\cp(10)=11$ & $\cp(17) = 18$\\
 \multicolumn{2}{l}{$\cp(14)=\cp(19) = \cp(20) = 21$} 
\end{tabular} 

\noindent Node $0$ and node $1$ are top-level $\push$.

\medskip
\noindent
Further, $\cs(5)={}$
\includegraphics[scale=.65,page=12]{tikz-pics} and
$\cs(10) = \cs(20) = \epsilon$. \\
Notice that the context-suffixes of 3 and 14 are connected to position 0.
This is because the corresponding $\push$ at position 1 is not the left-most.

\medskip
\noindent
Now, if we remove the $\push$ at position 0, its corresponding pops 5, 10, 20 
and their context-suffixes we obtain 

\medskip
\noindent
\includegraphics[scale=1,page=13]{tikz-pics}
 
\medskip
\noindent Then $\cs(3) = {}$
\includegraphics[scale=.65,page=14]{tikz-pics}
and $\cs(14) = {}$
\includegraphics[scale=.65,page=15]{tikz-pics} are not connected to the rest of 
the \twonw.

\clearpage
\section{Tree interpretation and decision procedures}
\label{sec:tree-decision}

In this section, we show that \twonws of split-width at most $k$ have special
tree-width (\STW) at most $2k$.  We deduce that \twonws of bounded split-width can
be interpreted in special tree terms (\STTs), which are binary trees denoting
graphs of bounded STW.
Special tree-width and special tree terms were introduced by Courcelle in 
\cite{Courcelle10}.

A crucial step towards our decision procedures is then to construct a tree
automaton \Aswk which accepts special tree terms denoting graphs that are \twonw
of split-width at most $k$.  The main difficulty is to make sure that the edge
relations $\NestRelOne$ and $\NestRelTwo$ of the graph are well nested.  To
achieve this with a tree automaton of size $2^{\poly(k)}$, we use the
characterization given by the \PDL formula $\phi_\mathsf{wn}$ of
Example~\ref{ex:matchrel-consitent}.  Similarly, we can construct a tree
automaton \Apopb of size $2^{\poly(\bound)}$ accepting \STTs denoting nested
words in $\pbnw$.

\subparagraph{Special tree terms} form an algebra to define graphs.
A $(\Sigma,\Gamma)$-labelled graph is a tuple
$G=\tuple{V,(E_\gamma)_{\gamma\in\Gamma},\lambda}$ where $\lambda\colon V\to\Sigma$ 
is the vertex labelling and $E_\gamma \subseteq V^2$ is the set of edges for 
each label $\gamma\in\Gamma$. For \twonw, we have three types of edges, so 
$\Gamma=\{\gonext,\NestRelOne,\NestRelTwo\}$. The syntax of \kSTTs 
over $(\Sigma,\Gamma)$ is given by
$$
\stt = (i,a) \mid \add{i}{j}{\gamma} \stt \mid \forget{i} \stt \mid
\rename{i}{j} \stt \mid \stt \sttunion \stt 
$$
where $a \in \Sigma$, $\gamma \in \Gamma$ and $i,j\in[k]=\{1,\ldots,k\}$ are
colors.  

Each \kSTT represents a colored graph $\sem\stt=(G_\tau,\chi_\tau)$ where
$G_\tau$ is a $(\Sigma,\Gamma)$-labelled graph
and $\chi_\tau\colon [k]\to V_\tau$ is a partial injective function assigning a vertex of
$G_\tau$ to some colors.
$\sem{(i,a)}$ consists of a single $a$-labelled vertex with color $i$.
$\add{i}{j}{\gamma}$ adds a $\gamma$-labelled edge to the vertices colored $i$ 
and $j$ (if such vertices exist). $\forget{i}$ removes color $i$ and 
$\rename{i}{j}$ exchanges the colors $i$ and $j$. Finally, $\sttunion$ 
constructs the disjoing union of the two graphs provided they use different 
colors. This operation is undefined otherwise.
The special tree-width of a graph $G$ is the least $k$ such that 
$G=G_\tau$ for some $(k+1)$-\STT $\stt$.

For instance, atomic split-\twonws are denoted by \STTs of the following form
\begin{itemize}
  \item $(1,a)$ for an internal event labelled $a$,

  \item $\add{1}{2}{\NestRelTwo}((1,a)\sttunion(2,b))$ for an order-2 matching 
  pair, and

  \item
  $\add{1}{2}{\NestRelOne}\cdots\add{1}{p}{\NestRelOne}((1,a_1)\sttunion\cdots\sttunion(p,a_p))$
  for order-1 push with $p-1$ pops.  
\end{itemize}
We call these \STTs \emph{atomic}.

To each split-tree $T$ of width $k$ with root labelled $\snw$, we associate a
$2k$-\STT $\stt$ such that $\sem{\stt}=(\snw,\chi)$ and all endpoints of factors
of $\snw$ have different colors.  Since we have at most $k$ factors, we may use
at most $2k$ colors.  A leaf of $T$ is labelled with an atomic split-\twonw and
we associate the corresponding atomic \STT as defined above.  At a binary node,
assuming that $\tau_\ell$ and $\tau_r$ are the \STTs of the children, we first
define $\tau'_r$ by renaming colors in $\tau_r$ so that colors in $\tau_\ell$
and $\tau'_r$ are disjoint, then we let $\tau=\tau_\ell\sttunion\tau'_r$.  At a
unary node $x$ with child $x'$, some factors of the spilt-\twonw
$\overline{\nw}_x$ are split resulting in the split-\twonw
$\overline{\nw}_{x'}$.  Assume that factor $u$ of $\overline{\nw}_x$ is split in
two factors $u'$ and $u''$ of $\overline{\nw}_{x'}$.  The right and left
endpoints of $u'$ and $u''$ respectively are colored, say by $i$ and $j$, in the
\STT $\tau'$ associated with $x'$.  Then, we add a successor edge
($\add{i}{j}{\gonext}$) and we forget $i$ if $|u'|>1$ and $j$ if $|u''|>1$.
We proceed similarly if a factor of $\overline{\nw}_x$ is split in
more than two factors of $\overline{\nw}_{x'}$, and we iterate for each factor 
of $\overline{\nw}_x$ which is split in $\overline{\nw}_{x'}$.

\begin{restatable}{proposition}{rstpropApopb}\label{prop:Apopb}
  There is a tree automaton \Apopb of size $2^{\poly(\bound)}$ accepting \kSTTs
  ($k=4\bound+4)$ and such that $\pbnw=\{G_\tau\mid 
  \tau\in\mathcal{L}(\Apopb)\}$.
\end{restatable}

The automaton \Apopb will accept precisely those \kSTTs arising from split-trees
as described above.  The construction of \Apopb is given below.  The main
difficulty is to check that the graph denoted by a special tree term denotes a
valid \twonw: ${<}={\gonext}^+$ should be a total order and the relations
$\NestRelOne$ and $\NestRelTwo$ should be well-nested.  Since we start from
nested words in $\pbnw$, we obtain split-trees of width at most $2\bound+2$ by
Theorem~\ref{thm:swofpopbounded}.  Notice that \Apopb needs not accept
\emph{all} \kSTTs denoting graphs that are nested words in $\pbnw$.

\begin{proof}
  First, we construct a tree automaton \Aword whose states constists of
  \begin{itemize}
    \item $n \leq 2 \beta +2$ : Number of factors in the split-\twonw.
    \item $c_\ell,c_r:[n] \to [k]$ : Functions describing colors of the left and right endpoints of 
    the factor.
  \end{itemize}
  To check more easily absence of $\gonext$-cycles, factors are numbered
  according to their guessed ordering in the final \twonw.  The transitions of
  \Aword ensure the following conditions:
  \begin{itemize}
    \item Atomic \STTs: the automaton checks that they denote atomic split 
    nested words in $\pbnw$. Then, the number $n$ of factors is at most $\beta +1$ and 
    $c_\ell,c_r$ are the identity maps $\mathsf{id}\colon [n]\to[k]$.
    
    \item Consider a subterm $\stt = \stt_1 \sttunion \stt_2$, we have $n = n_1 + n_2$. We guess how factors of $\stt_1$ and $\stt_2$ will be shuffled on each process and we inherit $c_\ell$ and $c_r$ accordingly.
    
    \item Consider a subterm $\stt = \add{i}{j}{\rightarrow}(\stt')$. 
    Let $(n',c'_\ell,c'_r)$ be the state at $\stt'$. We check that
    there are factors $x,y \in [n']$ such that $c'_r(x)=i$, $c'_\ell(y)=j$ and 
    $y=x+1$ (this checks that the guessed ordering of the factors is correct).
    The states at $\stt$ is easy to compute.
    \begin{itemize}
      \item $n=n'-1$
      \item $c_\ell(z) =c'_\ell(z)$ for $z \leq x$ and $c_\ell(z) = c'_\ell(z+1)$ for $z>x$ 
      \item $c_r(z) = c'_r(z)$ for $z < x$ and $c_r(z) = c'_r(z+1)$ for $z \geq x$ 
    \end{itemize}
    
    \item $\forget{i} \stt$: Check that
    $i\notin\mathsf{Im}(c_\ell)\cup\mathsf{Im}(c_r)$ is not in the image of the
    mappings $c_\ell$ and $c_r$.  We always keep the colors of the
    endpoints of the factors.
    
    \item $\rename{i}{j}$: Update $c_\ell$ and $c_r$ accordingly.
    
    \item Root: Check that n=1 (a single factor).
  \end{itemize}
  
  When an \STT $\tau$ is accepted by \Aword, then the relation $\gonext$ of the
  graph $G_\tau$ defines a total order on the vertices.  Hence, we have an
  underlying word with some nesting relations $\NestRelOne$ and $\NestRelTwo$.
  But we did not check that these relations are well-nested.  To check this
  property, we use the LCPDL formula $\phi_\mathsf{wn}$ of
  Example~\ref{ex:matchrel-consitent}.
  
  Consider an \STT $\stt$ accepted by $\mathcal{A}^{\beta}_{\ensuremath{word}}$,
  let $\sem{\stt} = (G_\tau,\chi_\tau)$ where
  $G_\tau=(V,\rightarrow,\NestRelOne,\NestRelTwo,\lambda)$.  we know that
  $(V,\rightarrow,\lambda)$ defines a word in $\Sigma^+$.  The graph $G_\tau$
  can be interpreted in $\stt$: we can build walking automata of size $\poly(k)$
  for $\rightarrow$, $\NestRelOne$, $\NestRelTwo$ and their converse.  Hence we
  can build an alternating two-way tree automaton of size $\poly(k)$ checking
  $\phi_\mathsf{wn}$.  We obtain an equivalent normal tree automaton $\Anesting$
  of size $2^{\poly(k)}$
  
  The final tree automaton is $\Apopb=\Aword\cap\Anesting$.
\end{proof}

\begin{restatable}{proposition}{rstpropAhopds}\label{prop:Ahopds}
  For each \twocpds $\hopds$ we can construct a tree automaton \Ahopds of
  size $2^{\poly(\bound,\sizeof{\hopds})}$ such that
  $\mathcal{L}(\Apopb\cap\Ahopds)=\{\tau\in\mathcal{L}(\Apopb)\mid 
  G_\tau\in\mathcal{L}(\hopds)\}$.
\end{restatable}

\newcommand{\snwleaf}{\raisebox{-.5mm}{\includegraphics[scale=1,page=24]{tikz-pics}}}

\begin{proof}
  The tree automaton \Ahopds essentially guesses the transitions of the \twocpds
  and checks that they form an accepting run.  To this end, we first construct a
  tree automaton reading \STTs whose leaves are additionally labelled with
  transitions of $\hopds$.  Then, we project away these additional labels to
  obtain the automaton \Ahopds.
  
  Consider an atomic \STT describing a \twonw of the form \snwleaf.  The tree
  automaton checks that the transitions labelling the leaves are of the form
  $(q_1^a,a,\pushs\Label,q_2^a)$, $(q_1^b,b,\pop,q_2^b)$, $(q_1^c,c,\pop,q_2^c)$
  and $(q_1^d,d,\pop,q_2^d)$.
  
  Another interesting case is for internal events carrying top tests.  We
  construct an alternating two-way walking tree automaton (A2A) of size 
  $\poly(\bound,\sizeof{\hopds})$,
  which visits all leaves of the input \STT. For each leaf $x$ which is labelled
  with a top test transition $(q_1,a,\toptest(\Label),q_2)$, the A2A walks to
  the leaf $y$ following the \PDL formula $\pi_{\mypush}$ of
  Example~\ref{ex:matchrel-consitent}. The top symbol of the order-2 stack at 
  node $x$ was pushed at node $y$. Hence, the A2A checks that the transition
  labelling $y$ is of the form $(q'_1,b,\pushs\Label,q'_2)$.
  This A2A is transformed into a classical tree automaton of size 
  $2^{\poly(\bound,\sizeof{\hopds})}$.
  
  It remains to check that, when following the linear order $\gonext$, the
  transitions labelling the leaves form an accepting run.  To this end, the
  bottom-up tree automaton remembers; for every factor of the split-\twonw
  associated with a node of the \STT, the starting control state and the ending
  control state.  It is then easy to verify when adding a $\gonext$ edge between
  two factors $u$ and $v$ that the target state of factor $u$ is the source
  state of factor $v$.  Finally, at the root of the \STT, the automaton accepts
  if there is only one factor and its source/target state is initial/final.
  
  The size of the tree automaton \Ahopds is $2^{\poly(\bound,\sizeof{\hopds})}$.
\end{proof}

We deduce from Proposition~\ref{prop:Ahopds} that non-emptiness checking of
\twocpds with respect to $\pbnw$ is in \textsc{ExpTime}.

\begin{restatable}{proposition}{rstpropApdl}\label{prop:Apdl}
  For each \PDL formula $\phi$ we can construct a tree automaton \Aphi of
  size $2^{\poly(\bound,\sizeof{\phi})}$ such that
  $\mathcal{L}(\Apopb\cap\Aphi)=\{\tau\in\mathcal{L}(\Apopb)\mid 
  G_\tau\in\mathcal{L}(\phi)\}$.
\end{restatable}

\begin{proof}
  The idea is to translate the \PDL formula $\phi$ to an alternating two-way tree
  automaton (A2A) of size $\poly(\bound,\sizeof{\phi})$.  Due to the specific form of the
  \STTs accepted by \Apopb, it is easy to encode the nesting relations
  $\NestRelOne$ and $\NestRelTwo$ with a walking automaton.  We can also easily
  build a walking automaton for the successor relation $\gonext$ by tracking the
  colors until we reach a node labelled $\add{i}{j}{\gonext}$.  One main
  difficulty is to cope with loops of \PDL.
  Here we use the result of \cite{Goeller2009} for PDL with converse and
  intersection.  In general, this can cause an exponential blow-up in the size of
  the A2A. But loop is a special case with bounded intersection-width and hence
  still allows a polynomial sized A2A.
  Finally, the A2A for $\phi$ is translated to the normal tree automaton \Aphi, 
  causing an exponential blow-up.
\end{proof}
  
We deduce that the bounded-pop satisfiability problem of \PDL can be solved in 
exponential time by checking emptiness of $\Apopb\cap\Aphi$.
Also, the bounded-pop model checking problem of \twocpds against \PDL can be
solved in exponential time by checking emptiness of 
$\Apopb\cap\Ahopds\cap\Anotphi$.

\medskip%
Similarly, for each \MSONW formula $\phi$, we can construct a tree automaton
\Aphi such that $\mathcal{L}(\Apopb\cap\Aphi)=\{\tau\in\mathcal{L}(\Apopb)\mid
G_\tau\in\mathcal{L}(\phi)\}$.  We deduce the decidability of the bounded-pop
satisfiability problem of \MSONW and of the bounded-pop model checking problem
of \twocpds against \MSONW. This concludes the proof of Theorem~\ref{thm:main}.

\end{document}